\documentclass[11pt]{article}

\usepackage[utf8]{inputenc}

\usepackage{amsmath}
\usepackage{amsthm}
\usepackage{amssymb}
\usepackage{graphicx}
\usepackage{comment}
\usepackage{algorithmicx}
\usepackage{algpseudocode}
\usepackage{algorithm}
\usepackage{subfig}
\usepackage[table]{xcolor}
\usepackage{booktabs}

\DeclareUnicodeCharacter{1EC5}{} 

\usepackage{color}   
\usepackage{hyperref}
\hypersetup{
    colorlinks=true, 
    linktoc=all,     
    linkcolor=blue,
}

\usepackage{todonotes}

\usepackage[
backend=biber,
style=alphabetic,
sorting=ynt
]{biblatex}
\newtheorem{thm}{Theorem}
\newtheorem{corollary}[thm]{Corollary}
\newtheorem{definition}[thm]{Definition}
\newtheorem{lem}[thm]{Lemma}
\newtheorem{prop}[thm]{Proposition}
\newtheorem{rmk}[thm]{Remark}

\newcommand{\poly}{\text{poly}}

\newcommand{\norm}[2]{\left\|#1\right\|_{#2}}
\newcommand{\inner}[2]{\left\langle #1, #2 \right\rangle}
\newcommand{\abs}[1]{\left|#1\right|}
\newcommand{\E}{\mathbb{E}}
\newcommand{\pr}{\mathbb{P}}
\newcommand{\R}{\mathbb{R}}
\newcommand{\parens}[1]{\left(#1 \right)}

\newcommand{\eps}{\epsilon}
\newcommand{\tildeO}{\widetilde{O}}
\newcommand{\tildeTheta}{\widetilde{\Theta}}
\newcommand{\tildeOmega}{\widetilde{\Omega}}
\newcommand{\pinv}{\dagger}

\newcommand{\ptester}[1]{$\ell_{#1}$\text{-tester}}
\newcommand{\psdcone}[1]{\Delta^{#1}_{+}}

\newcommand{\True}{\text{True}}
\newcommand{\False}{\text{False}}

\renewcommand{\E}{\mathbb{E}}

\DeclareMathOperator{\Tr}{Tr}
\DeclareMathOperator{\diag}{diag}
\DeclareMathOperator{\rank}{rk}
\DeclareMathOperator{\argmin}{argmin}
\DeclareMathOperator{\spn}{span}
\DeclareMathOperator{\var}{Var}

\usepackage[margin=1in]{geometry}

\newcommand{\td}[1]{\todo[inline]{#1}}
\renewcommand{\td}[1]{}

\title{Testing Positive Semidefiniteness Using Linear Measurements 
\thanks{Authors DN and WS were partially supported by NSF DMS $\#$2011140 and NSF DMS $\#$2108479. DW was supported by NSF CCF $\#$1815840 and Office of Naval Research Grant N00014-18-1-2562.}
}

\author{Deanna Needell\\UCLA 
\and
William Swartworth
\\UCLA \and
David P. Woodruff\\CMU
}

\date{}

\begin{document}

\maketitle

\thispagestyle{empty}

\begin{abstract}
\td{Check the matrix-vector bounds for $p=2$}

We study the problem of testing whether a symmetric $d \times d$ input matrix $A$ is symmetric positive semidefinite (PSD), or is $\epsilon$-far from the PSD cone, meaning that $\lambda_{\min}(A) \leq - \epsilon \|A\|_p$, where $\|A\|_p$ is the Schatten-$p$ norm of $A$. In applications one often needs to quickly tell if an input matrix is PSD, and a small distance from the PSD cone may be tolerable. We consider two well-studied query models for measuring efficiency, namely, the matrix-vector and vector-matrix-vector query models. We first consider one-sided testers, which are testers that correctly classify any PSD input, but may fail on a non-PSD input with a tiny failure probability. Up to logarithmic factors, in the matrix-vector query model we show a tight $\tildeTheta(1/\epsilon^{p/(2p+1)})$ bound, while in the vector-matrix-vector query model we show a tight $\tildeTheta(d^{1-1/p}/\epsilon)$ bound, for every $p \geq 1$. We also show a strong separation between one-sided and two-sided testers in the vector-matrix-vector model, where a two-sided tester can fail on both PSD and non-PSD inputs with a tiny failure probability. In particular, for the important case of the Frobenius norm, we show that any one-sided tester requires $\tildeOmega(\sqrt{d}/\epsilon)$ queries.  However we introduce a bilinear sketch for two-sided testing from which we construct a Frobenius norm tester achieving the optimal $\tildeO(1/\epsilon^2)$ queries. We also give a number of additional separations between adaptive and non-adaptive testers. Our techniques have implications beyond testing, providing new methods to approximate the spectrum of a matrix with Frobenius norm error using dimensionality reduction in a way that preserves the signs of eigenvalues. 
\end{abstract}

\clearpage
\setcounter{page}{1}

\section{Introduction}
A real-valued matrix $A \in \mathbb{R}^{n \times n}$ is said to be Positive Semi-Definite (PSD) if it defines
a non-negative quadratic form, namely, if
$x^T A x \geq 0$ for all $x$. If $A$ is symmetric, the setting on which we focus,
this is equivalent to the eigenvalues of $A$ being non-negative. Multiple works \cite{ks03,han2017approximating,bakshi2020testing} have studied the problem of testing whether a real matrix is PSD, or is far from being PSD, and this testing problem has numerous applications, including to faster algorithms for linear systems and linear algebra problems, detecting the existence of community structure, ascertaining local convexity, and differential equations; we refer the reader to \cite{bakshi2020testing} and the references therein. 

We study two fundamental query models.  In the matrix-vector model, one is given implicit access to a matrix $A$ and may query $A$ by choosing a vector $v$ and receiving the vector $Av.$  In the vector-matrix-vector model one chooses a pair of vectors $(v,w)$ and queries the bilinear form associated to $A$. In other words the value of the query is $v^T A w$. In both models, multiple, adaptively-chosen queries can be made, and the goal is to minimize the number of queries to solve a certain task. These models are standard computational models in the numerical linear algebra community, see, e.g., \cite{han2017approximating} where PSD testing was studied in the matrix-vector query model. These models were recently formalized in the theoretical computer science community in \cite{sun2019querying,rashtchian2020vector}, though similar models have been studied in numerous fields, such as the number of measurements in compressed sensing, or the sketching dimension of a streaming algorithm.  The matrix-vector query and vector-matrix-vector query models are particularly relevant when the input matrix $A$ is not given explicitly.

A natural situation occurs when $A$ is presented implicitly as a the Hessian of a function $f: \R^d \rightarrow \R^d$ at a point $x_0$, where $f$ could be the loss function of a neural network for example. One might want to quickly distinguish between a proposed optimum of $f$ truly being a minimum, or being a saddle point with a direction of steep downward curvature.  Our query model is quite natural in this context. A Hessian-vector product is efficient to compute using automatic differentiation techniques.  A vector-matrix-vector product corresponds to a single second derivative computation, $D^2 f(v,w)$. This can be approximated using $4$ function queries by the finite difference approximation
$D^2 f(v,w) \approx \frac{f(x_0 + hv + hw) - f(x_0+hv) - f(x_0 + hw) + f(x_0)}{h^2},$
where $h$ is small.

While there are numerically stable methods for computing the spectrum of a symmetric matrix, and thus determining if it is PSD, these methods can be prohibitively slow for very large matrices, and require a large number of matrix-vector or vector-matrix-vector products. Our goal is to obtain significantly more efficient algorithms in these models, and we approach this problem from a property testing perspective. In particular, we focus on the following version of the PSD-testing problem. In what follows, $\|A\|_p = \left (\sum_{i=1}^n \sigma_i^p \right )^{1/p}$ is the Schatten-$p$ norm of $A$, where the $\sigma_i$ are the  singular values of $A$. 
\begin{definition}
\label{def:tester}
For $p\in [1,\infty]$, an $(\epsilon,\ell_p)$-tester is an algorithm that makes either matrix-vector or vector-matrix-vector queries to a real symmetric matrix $A$, and outputs \True\,with at least $2/3$ probability if $A$ is PSD, and outputs \False\,with $2/3$ probability if $A$ is $\epsilon\norm{A}{p}$-far in spectral distance from the PSD cone, or equivalently, if the minimum eigenvalue $\lambda_{\min}(A) \leq -\eps \norm{A}{p}$. If the tester is guaranteed to output \True\,on all PSD inputs (even if the input is generated by an adversary with access to the random coins of the tester), then the tester has one-sided error.  Otherwise it has two-sided error. When $\epsilon$ is clear from the context we will often drop the $\epsilon$ and simply refer to an  \ptester{p}. 
\end{definition}

Our work fits more broadly into the growing body of work on property testing for linear algebra problems, see, for example
\cite{balcan2019testing, bakshi2020testing,bmr21}. However, a key difference is that we focus on matrix-vector and vector-matrix-vector query models, which might be more appropriate than the model in the above works which charges a cost of $1$ for reading a single entry. Indeed, such models need to make the assumption that the entries of the input are bounded by a constant or slow-growing function of $n$, as otherwise strong impossibility results hold. This can severely limit the applicability of such algorithms to real-life matrices that do not have bounded entries; indeed, even a graph Laplacian matrix with a single degree that is large would not fit into the above models. 
In contrast, we use the matrix-vector and vector-matrix-vector models, which are ideally suited for modern machines such as graphics processing units and when the input matrix cannot fit into RAM, and are standard models in scientific computing, see, e.g., \cite{bai1996some}. 

While we focus on vector-matrix-vector queries, our results shed light on several other natural settings. Many of our results are in fact tight for general linear measurements which vectorize the input matrix and apply adaptively chosen linear forms to it.   
For long enough streams the best known single or multi-pass algorithms {\it for any problem} in the turnstile streaming model form a sketch using general linear measurements, and with some additional restrictions, it can be shown that the optimal multi-pass streaming algorithm just adaptively chooses general linear measurements \cite{ai2016new}. Therefore, it is quite plausible that many of our vector-matrix-vector algorithms give tight single pass streaming bounds, given that vector-matrix-vector queries are a special case of general linear measurements, and that many our lower bounds are tight even for general linear measurements. 

Moreover our vector-matrix-vector algorithms lead to efficient communication protocols for deciding whether a distributed sum of matrices is PSD, provided that exact vector-matrix-vector products may be communicated.  While we expect our methods to be stable under small perturbations (i.e. when the vector-matrix-vector products are slightly inexact), we leave the full communication complexity analysis to future work.

We note that our PSD-testing problem is closely related to that of approximating the largest eigenvalue of a PSD matrix. Indeed by appropriately negating and shifting the input matrix, it is essentially equivalent to estimate the largest eigenvalue of a PSD matrix $A$ up to additive error $\eps \parens{\sum_i\abs{\lambda_{\max}(A) - \lambda_i(A)}^p}^{1/p}.$  However this problem is much less natural as real-world matrices often have many small eigenvalues, but only a few large eigenvalues.

\subsection{Our Contributions}\label{sec:cont}

{
\renewcommand{\arraystretch}{1.5}

\td{check table for typos}
\begin{table}
\label{table_of_results}

\begin{center}
\begin{tabular}{ |p{5.8cm}||p{3.9cm}|p{5cm}| }
 \hline
 
 \rowcolor{lightgray}
 \multicolumn{3}{|c|}{\textbf{Vector-matrix-vector queries}} \\
 \hline
 Adaptive, one-sided $\ell_p$ & $\tildeTheta(\frac1\eps d^{1-1/p})$  & Corollary~\ref{cor:adaptive_vmv_upper_p}, Theorem~\ref{thm:one_sided_p_lower_bound} \\
 Non-adaptive, one-sided $\ell_p$ & $\tildeTheta(\frac{1}{\eps^2} d^{2 - 2/p})$ & Corollary~\ref{cor:andoni_nguyen_corollary_p}, Theorem~\ref{thm:nonadaptive_one_sided_vmv_lower} \\
 Adaptive, two-sided $\ell_2$ & $\tildeTheta(\frac{1}{\eps^2})^*$ & Proposition~\ref{thm:adaptive_l2_upper}, Corollary~\ref{cor:adaptive_two_sided_lower_bound} \\
 Non-adaptive, two-sided $\ell_2$ & $\tildeTheta(\frac{1}{\eps^{4}})^*$ & Theorem~\ref{thm:bilinear_sketch_guarantee}, Theorem~\ref{thm:nonadaptive_two_sided_vmv_lower_bound} \\
 Adaptive, two-sided $\ell_p,$ $2\leq p<\infty$ & $\tildeTheta(\frac{1}{\eps^2} d^{1-2/p})^*$ & Corollary~\ref{cor:adaptive_lp_upper}, Corollary~\ref{cor:adaptive_two_sided_lower_bound}\\
 
 \hline
 \rowcolor{lightgray}
 \multicolumn{3}{|c|}{\textbf{Matrix-vector queries}} \\
 \hline
Adaptive one-sided $\ell_p$ & $\tildeO((1/\eps)^{p/(2p+1)} \log d)$, $\Omega((1/\eps)^{p/(2p+1)})$  & Theorem~\ref{thm:adaptive_mv_upper}, Theorem~\ref{thm:adaptive_mv_lower} \\
Adaptive one-sided $\ell_1$ & $\tildeTheta((1/\eps)^{1/3})$ & Theorem~\ref{thm:adaptive_mv_upper}, Theorem~\ref{thm:adaptive_mv_lower} \\
 Non-adaptive one-sided $\ell_p$ & $\Theta(\frac{1}{\eps}d^{1-1/p})$ & Proposition~\ref{prop:nonadaptive_one_sided_mv_upper}, Corollary~\ref{cor:mv_one_sided_non_adaptive_lower} \\

 \hline
\end{tabular}
\caption{Our upper and lower bounds for the matrix-vector and vector-matrix-vector query models.  ${}^{*}$ indicates that the lower bound holds for general linear measurements.}
\end{center}
\end{table}
}


We study PSD-testing in the matrix-vector and vector-matrix-vector models.  In particular, given a real symmetric matrix $A,$ and $p\in [1,\infty]$, we are interested in deciding between (i) $A$ is PSD and (ii) $A$ has an eigenvalue less than $-\eps\norm{A}{p},$ where $\norm{A}{p}$ is the Schatten $p$-norm of $A.$

\paragraph{Tight Bounds for One-sided Testers.} We make particular note of the distinction between one-sided and two-sided testers. In some settings one is interested in a tester that produces one-sided error.  When such a tester outputs $\False$, it must be able to produce a proof that $A$ is not PSD.  The simplest such proof is a witness vector $v$ such that $v^T A v < 0$, and indeed we observe that in the matrix-vector model, any one-sided tester can produce such a $v$ when it outputs $\False.$  This may be a desirable feature if one wishes to apply these techniques to saddle point detection for example: given a point that is not a local minimum, it would be useful to produce a descent direction so that optimization may continue.  In the vector-matrix-vector model the situation is somewhat more complicated in general, but all of our one-sided testers produce a witness vector whenever they output \False.

We provide {\it optimal bounds} for one-sided testers for both matrix-vector and vector-matrix-vector models. The bounds below are stated for constant probability algorithms. Here $\tildeO(f) = f \cdot \poly(\log f)$.


\begin{enumerate} 

\item In the matrix-vector query model, we show that up to a factor of $\log d$, $\tildeTheta(1/\epsilon^{p/(2p+1)})$ queries are necessary and sufficient for an $\ell_p$-tester for any $p \geq 1$.  In the $p=1$ case, we note that the $\log d$ factor may be removed. 

\item In the vector-matrix-vector query model, we show that $\tildeTheta(d^{1-1/p}/\epsilon)$ queries are necessary and sufficient for an $\ell_p$-tester for any $p \geq 1$. Note that when $p = 1$ we obtain a very efficient $\tildeO(1/\epsilon)$-query algorithm. In particular, our tester for $p = 1$ has query complexity independent of the matrix dimensions, and we show a sharp phase transition for $p > 1$, showing in some sense that $p = 1$ is the largest value of $p$ possible for one-sided queries. 
\end{enumerate} 

The matrix-vector query complexity is very different than the vector-matrix-vector query complexity, as the query complexity is $\poly(1/\epsilon)$ for any $p \geq 1$, which captures the fact that each matrix-vector query response reveals more information than that of a vector-matrix-vector query, though a priori it was not clear that such responses in the matrix-vector model could not be compressed using vector-matrix-vector queries. 

\paragraph{An Optimal Bilinear Sketch for Two-Sided Testing.}

Our main technical contribution for two-sided testers is a bilinear sketch for PSD-testing with respect to the Frobenius norm, i.e. $p=2$.  We consider a Gaussian sketch $G^T A G$, where $G$ has small dimension $\tildeO(\frac{1}{\eps^2}).$  By looking at the smallest eigenvalue of the sketch, we are able to distinguish between $A$ being PSD and being $\eps$-far from PSD.  Notably this tester may reject even when $\lambda_{\min}(G^T A G) > 0$, which results in a two-sided error guarantee. This sketch allows us to obtain tight two-sided bounds in the vector-matrix-vector model for $p\geq 2$, both for adaptive and non-adaptive queries. 


\paragraph{Separation Between One-Sided and Two-Sided Testers.}  Surprisingly, we show a separation between one-sided and two-sided testers in the vector-matrix-vector model. For the important case of the Frobenius norm, i.e., $p=2$, we utilize our bilinear sketch to construct a $\tildeO(1/\eps^2)$ query two-sided tester, whereas by our results above, any adaptive one-sided tester requires at least $\Omega(\sqrt{d}/\epsilon)$ queries.  

We also show that for any $p > 2$, any possibly adaptive two-sided tester requires $d^{\Omega(1)}$ queries for constant $\epsilon$, and thus in some sense, $p = 2$ is the largest value of $p$ possible for two-sided queries.  

\paragraph{On the Importance of Adaptivity.}
We also study the role of adaptivity in both matrix-vector and vector-matrix-vector models.
In both the one-sided and two-sided vector-matrix-vector models we show a quadratic separation between adaptive and non-adaptive testers, which is the largest gap possible for any vector-matrix-vector problem \cite{sun2019querying}.


In the matrix-vector model, each query reveals more information about $A$ than in the vector-matrix-vector model, allowing for even better choices for future queries.  Thus we have an even larger gap between adaptive and non-adaptive testers in this setting. 


\paragraph{Spectrum Estimation.}
While the two-sided tester discussed above yields optimal bounds for PSD testing, it does not immediately give a way to estimate the negative eigenvalue when it exists.  Via a different approach, we show how to give such an approximation with $\eps \norm{A}{F}$ additive error.  In fact, we show how to approximate all of the top $k$ eigenvalues of $A$ using $O(k^2 \poly\frac{1}{\eps})$ non-adaptive vector-matrix-vector queries, which may be of independent interest.

We note that this gives an $O(k^2 \poly\frac{1}{\eps})$ space streaming algorithm for estimating the top $k$ eigenvalues of $A$ to within additive Frobenius error. Prior work yields a similar guarantee for the singular values \cite{andoni2013eigenvalues}, but cannot recover the signs of eigenvalues.

\subsection{Our Techniques}
\paragraph{Matrix-Vector Queries.}
For the case of adaptive matrix-vector queries, we show that Krylov iteration starting with a single random vector yields an optimal \ptester{p} for all $p$. Interestingly, our analysis is able to beat the usual Krylov matrix-vector query bound for approximating the top eigenvalue, as we modify the usual polynomial analyzed for eigenvalue estimation to implicitly implement a {\it deflation} step of all eigenvalues above a certain threshold. We do not need to explicitly know the values of the large eigenvalues in order to deflate them; rather, it suffices that there exists a low degree polynomial in the Krylov space that implements this deflation.

We note that this idea of implicitly deflating the top eigenvalues first appeared in \cite{axelsson1986rate}.  More recently this observation was applied in a setting very similar to ours by \cite{spielman2009note} who deflated eigenvalues that are large relative to the trace.

Further, we show that this technique is tight for all $p \geq 1$ by showing that any smaller number of matrix-vector products would violate a recent lower bound of \cite{braverman2020gradient} for approximating the smallest eigenvalue of a Wishart matrix. This lower bound applies even to two-sided testers.

\paragraph{Vector-Matrix-Vector Queries.} We start by describing our result for $p = 1$. We give one of the first examples of an algorithm in the vector-matrix-vector query model that leverages adaptivity in an interesting way. Most known algorithms in this model work non-adaptively, either by applying a bilinear sketch to the matrix, or by making many independent queries in the case of Hutchinson's trace estimator  \cite{hutchinson1989stochastic}. Indeed, the algorithm of \cite{andoni2013eigenvalues} works by computing $G^TAG$ for a Gaussian matrix $G$ with $1/\epsilon$ columns, and arguing that all eigenvalues that are at least $\epsilon \|A\|_1$ can be estimated from the sketch. The issue with this approach is that it uses $\Omega(1/\epsilon^2)$ queries and this bound is tight for non-adaptive testers! One could improve this by running our earlier matrix-vector algorithm on top of this sketch, without ever explicitly forming the $1/\epsilon \times 1/\epsilon$ matrix $G^TAG$; however, this would only give an $O(1/\epsilon^{4/3})$ query algorithm. 

To achieve our optimal $\tildeO(1/\epsilon)$ complexity, our algorithm instead performs a novel twist to Oja's algorithm \cite{oja1982simplified}, the latter being a stochastic gradient descent (SGD) algorithm applied to optimizing the quadratic form $f(x) = x^T A x$ over the sphere. In typical applications, the randomness of SGD arises via randomly sampling from a set of training data.  In our setting, we instead  artificially introduce randomness at each step, by computing the projection of the gradient onto a randomly chosen direction.  This idea is implemented via the iteration \begin{equation}
    x^{(k+1)} = x^{k} - \eta (g^T A x^{k})g \,\,\, \text{where}\,\, g\sim \mathcal{N}(0,I)
\end{equation}
for a well-chosen step size $\eta.$  If $f$ ever becomes negative before reaching the maximum number of iterations, then the algorithm outputs \False, otherwise it outputs \True.  For $p=1$, we show that this scheme results in an optimal tester (up to logarithmic factors). Our proof uses a second moment analysis to analyze a random walk, that is similar in style to \cite{jain2016streaming}, though our analysis is quite different. Whereas \cite{jain2016streaming} considers an arbitrary i.i.d. stream of unbiased estimators to $A$ (with bounded variance), our estimators are simply $g g^T A,$ which do not seem to have been considered before.  We leverage this special structure to obtain a better variance bound on the iterates throughout the first $\tildeO(1/\eps)$ iterations, where each iteration can be implemented with a single vector-matrix-vector query. Our algorithm and analysis gives a new method for the fundamental problem of approximating eigenvalues. 
    
    Our result for general $p > 1$ follows by relating the Schatten-$p$ norm to the Schatten-$1$ norm and invoking the algorithm above with a different setting of $\epsilon$. 
    We show our method is optimal by proving an $\Omega(d^{2-2/p}/\epsilon^2)$ lower bound for non-adaptive one-sided testers, and then using a theorem in \cite{rashtchian2020vector} which shows that adaptive one-sided testers can give at most a quadratic improvement. We note that one could instead use a recent streaming lower bound of \cite{inw22} to prove this lower bound, though such a lower bound would depend on the bit complexity.

\paragraph{Two-Sided Testers.}

The key technical ingredient behind all of our two-sided testers is a bilinear sketch for PSD-testing. Specifically, we show that a sketch of the form $G^T A G$ with $G\in R^{d\times k}$ is sufficient for obtaining a two-sided tester for $p=2$.  In contrast to the $p=1$ case, we do not simply output \False\,when $\lambda_{\min} := \lambda_{\min}(G^TAG) < 0$ as such an algorithm would automatically be one-sided.  Instead we require a criterion to detect when $\lambda_{\min}$ is suspiciously small.  For this we require two results.

The first is a concentration inequality for $\lambda_{\min}(G^TAG)$ when $A$ is PSD.  We show that $\lambda_{\min} \geq \Tr(A) - \tildeO(\sqrt{k})\norm{A}{F}$ with very good probability.  
This result is equivalent to bounding the smallest singular value of $A^{1/2}G$, which is a Gaussian matrix whose rows have different variances.  Although many similar bounds for constant variances exist in the literature \cite{litvak2005smallest,vershynin2018high}, we were not able to find a bound for general covariances.  In particular, most existing bounds do not seem to give the concentration around $\Tr(A)$ that we require.  


When $A$ has a negative eigenvalue of $-\eps$, we show that $\lambda_{\min} \leq \Tr(A) - \eps O(k)$.  By combining these two results, we are able to take $k=\tildeO(1/\eps^2)$, yielding a tight bound for non-adaptive testers in the vector-matrix-vector model. In fact this bound is even tight for general linear sketches, as we show by applying the results in \cite{li2016tight}.

We also utilize this bilinear sketch to give tight bounds for adaptive vector-matrix-vector queries, and indeed for general linear measurements.  By first (implicitly) applying the sketch, and then shifting by an appropriate multiple of the identity we are able to reduce to the $(\eps^2, \ell_1)$-testing problem, which as described above may solved using $\tildeO(1/\eps^2)$ queries.  



\paragraph{Spectrum Estimation.} A natural approach for approximating the eigenvalues of an $n \times n$ matrix $A$ is to first compute a sketch $G^T A G$ or a sketch $G^T A H$ for Gaussian matrices $G$ and $H$ with a small number of columns. Both of these sketches appear in \cite{andoni2013eigenvalues}. As noted above, $G^TAG$ is a useful non-adaptive sketch for spectrum approximation, but the error in approximating each eigenvalue is proportional to the Schatten-$1$ norm of $A$. One could instead try to make the error depend on the Frobenius norm $\|A\|_2$ of $A$ by instead computing $G^TA H$ for independent Gaussian matrices $G$ and $H$, but now $G^TA H$ is no longer symmetric and it is not clear how to extract the signs of the eigenvalues of $A$ from $G^TAH$. Indeed, \cite{andoni2013eigenvalues} are only able to show that the {\it singular values} of $G^TAH$ are approximately the same as those of $A$, up to additive $\epsilon \|A\|_2$ error. We thus need a new way to {\it preserve sign information of eigenvalues}. 


To do this, we show how to use results for providing the best PSD low rank approximation to an input matrix $A$, where $A$  need not be PSD and need not even be symmetric. In particular, in \cite{clarkson2017lowPSD} it was argued that if $G$ is a Gaussian matrix with $O(k/\epsilon)$ columns, then if one sets up the optimization problem $\min_{\textrm{rank k PSD} \ Y} \|AGYG^TA^T - A\|_F^2$, then the cost will be at most $(1+\epsilon)\|A_{k, +} - A\|_F^2$, where $A_{k, +}$ is the best rank-$k$ PSD approximation to $A$. By further sketching on the left and right with so-called {\it affine embeddings} $S$ and $T$, which have $\poly(k/\epsilon)$ rows and columns respectively, one can reduce this problem to $\min_{\textrm{rank k PSD}\ Y}\|SAGYG^TA^TT-SAT\|_F^2$, and now $SAG$, $G^TA^TT$ and $SAT$ are all $\poly(k/\epsilon) \times \poly(k/\epsilon)$ matrices so can be computed with a $\poly(k/\epsilon)$ number of vector-matrix-vector products. At this point the optimal $Y$ can be found with no additional queries and its cost can be evaluated.  By subtracting this cost from $\norm{A}{F}^2$, we approximate $\norm{A_{+,i}}{F}^2$, and $\norm{A_{-,i}}{F}^2$ for all $i\in [k]$, which in turn allows us to produce (signed) estimates for the eigenvalues of $A$.


When $A$ is PSD, we note that Theorem 1.2 in \cite{andoni2013eigenvalues} is able to reproduce our spectral approximation guarantee using sketching dimension $O(\frac{k^2}{\eps^8})$, compared to our sketch of dimension $O(\frac{k^2}{\eps^{12}})$.  However as mentioned above, our guarantee is stronger in that it allows for the signs of the eigenvalues to be recovered, i.e. our guarantee holds even when $A$ is not PSD.  Additionally, we are able to achieve $O(\frac{k^2}{\eps^8})$ using just a single round of adaptivity.

\paragraph{Lower Bounds for One-sided Testers.} To prove lower bounds for one-sided non-adaptive testers, we first show that a one-sided tester must be able to produce a witness whenever it outputs \False. In the matrix-vector model, the witness is a vector $v$ with $v^T A v < 0$, and in the vector-matrix-vector model, the witness is a PSD matrix $M$ with $\inner{M}{A}<0$.  In both cases we show that even for simplest non-PSD spectrum $(-\lambda,1,\ldots, 1)$, that it takes many queries to produce a witness when $\lambda$ is small.  In the matrix-vector model, our approach is simply to show that the $-\lambda$ eigenvector is typically far from span of all queried vectors, when the number of queries is small. This will imply that $A$ is non-negative on the queried subspace, which precludes the tester from producing a witness.  In the vector-matrix-vector model our approach is similar, however now the queries take the form of inner products against rank one matrices $x_ix_i^T.$  We therefore need to work within the space of symmetric matrices, and this requires a more delicate argument.

\subsection{Additional Related Work}
Numerous other works have considered matrix-vector queries and vector-matrix queries, see, e.g., \cite{meyer2021hutch++,braverman2020gradient,sun2019querying,simchowitz2018tight,musco2015randomized,wimmer2014optimal}. We outline a few core areas here.
%

\paragraph{Oja's Algorithm.}
Several works have considered Oja's algorithm in the context of streaming PCA, \cite{shamir2016convergence,jain2016streaming,allen2017first}. \cite{jain2016streaming} gives a tight convergence rate for iteratively approximating the top eigenvector of a PSD matrix, given an eigengap, and \cite{allen2017first} extends this to a gap free result for $k$-PCA. 

\paragraph{PSD Testing.} As mentioned above, PSD-testing has been investigated in the bounded entry model, where one assumes that the entries of $A$ are bounded by $1$ \cite{bakshi2020testing}, and one is allowed to query the entries of $A$. This is a restriction of the vector-matrix-vector model that we consider where only coordinate vectors may be queried.  
However since we consider a more general query model, we are able to give a better adaptive tester -- for us $\tildeO(1/\eps)$ vector-matrix-vector queries suffice, beating the $\Omega(1/\eps^2)$ lower bound given in \cite{bakshi2020testing} for entry queries.


Another work on PSD-testing is that of \cite{han2017approximating}, who construct a PSD-tester in the matrix-vector model. They first show how to approximate a general trace function $\sum f(\lambda_i)$ for sufficiently smooth $f$, by using a Chebyshev polynomial construction to approximate $f$ in the sup-norm over an interval. This allows them to construct an \ptester{\infty} by taking $f$ to be a smooth approximation of a shifted Heaviside function. Unfortunately this approach is limited to \ptester{\infty}s, and does not achieve the optimal bound; they require $\Omega((\log d)/\eps)$ matrix-vector queries compared to the $\tildeO((\log d)/\sqrt{\epsilon})$ queries achieved by Krylov iteration.

\paragraph{Spectrum Estimation.}
The closely-related problem of spectrum estimation has been considered several times, in the context of sketching the largest $k$ elements of the spectrum \cite{andoni2013eigenvalues} discussed above, and approximating the entire spectrum from entry queries in the bounded entry model \cite{bhattacharjee2021sublinear}.

\subsection{Notation}

A symmetric matrix $A$ is positive semi-definite (PSD) if all eigenvalues are non-negative. We use $\psdcone{d}$ to represent the PSD-cone, which is the subset of $d\times d$ symmetric matrices that are PSD.

For a matrix $A$ we use $\norm{A}{p}$ to denote the Schatten $p$-norm, which is the $\ell_p$ norm of the vector of singular values of $A.$  The Frobenius norm will play a special role in several places, so we sometimes use the notation $\norm{A}{F}$ to emphasize this. Additionally, $\norm{A}{}$ without the subscript indicates operator norm (which is equivalent to $\norm{A}{\infty}$). 

We always use $d$ to indicate the dimension of the matrix being tested, and use $\epsilon < 1$ to indicate the parameter in Definition~\ref{def:tester}.

When applied to vectors, $\inner{\cdot}{\cdot}$ indicates the standard inner product on $\R^n.$ When applied to matrices, it indicates the Frobenius inner product $\inner{X}{Y} := \Tr(X^T Y).$

 $S^{d-1}$ indicates the set of all unit vectors in $\R^d.$

We use the notation $X^{\pinv}$ to indicate the Moore-Penrose pseudoinverse of $X$.

For a symmetric matrix $A\in \R^{d\times d}$ with eigenvalues $\lambda_1 \geq \lambda_2 \geq \ldots \geq \lambda_d$, we let $A_k$ denote the matrix $A$ with all but the top $k$ eigenvalues zeroed out.  Formally, if $U$ is an orthogonal matrix diagonalizing $A$, then $A_k = U^T \diag(\lambda_1, \ldots, \lambda_k, 0, \ldots, 0) U.$  We also let $A_{-k} = A - {A_k}.$

Throughout, we use $c$ to indicate an absolute constant.  The value of $c$ may change between instances.


\section{Vector-matrix-vector queries}

\subsection{An optimal one-sided tester.}

To construct our vector-matrix-vector tester, we analyze the iteration 
\begin{equation}
    \label{oja_update}
    x^{(k+1)} = x^{(k)} - \eta \left((g^{(k)})^T A x^{(k)} \right)g^{(k)} 
    = \parens{I - \eta g^{(k)}(g^{(k)})^T A}x^{(k)},
\end{equation}
where $g^{(k)} \sim \mathcal{N}(0,I_d)$ and $x^{(0)}\sim \mathcal{N}(0,I_d).$

Our algorithm is essentially to run this scheme for a fixed number of iterations with with well-chosen step size $\eta.$ If the value of $(x^{(k)})^T A x^{}(k)$ ever becomes negative, then we output \False, otherwise we output \True. Using this approach we prove the following.  

\begin{thm}
\label{thm:adaptive_vmv_upper}
There exists a one-sided adaptive \ptester{1}, that makes $O(\frac1\eps \log^3 \frac1\eps)$ vector-matrix-vector queries to $A.$
\end{thm}

As an immediate corollary we obtain a bound for \ptester{p}s.

\begin{corollary}
\label{cor:adaptive_vmv_upper_p}
There is a one-sided adaptive \ptester{p} that makes $O(\frac1\eps d^{1-1/p} \log^3(\frac{1}{\eps}d^{1-1/p}))$ vector-matrix-vector queries.
\end{corollary}
\begin{proof}
This follows from the previous result along with the bound $\norm{A}{p} \geq d^{1/p -1} \norm{A}{1}.$
\end{proof}

We now turn to the proof of Theorem~\ref{thm:adaptive_vmv_upper}.  Since our iterative scheme is rotation-invariant, we assume without loss of generality that $A = \diag(\lambda_1, \ldots, \lambda_d).$   For now, we assume that $\norm{A}{1} \leq 1,$ and that the smallest eigenvalue of $A$ is $\lambda_1 = -\epsilon.$  We consider running the algorithm for $N$ iterations. We will show that our iteration finds an $x$ with $x^T A x < 0$ in $ N = \tilde{O}(1/\epsilon)$ iterations. We will use $c$ to denote absolute constants that we don't track, and that may vary between uses.

Our key technical lemma is to show that the first coordinate (which is associated to the $-\epsilon$ eigenvalue) grows fairly quickly with good probability.


\begin{lem}
\label{lem:first_coord_grows}
Suppose $\eta$ and $N$ satisfy the following list of assumptions: 
    (1) $\eta \leq \frac{1}{4}$, 
    (2) $\eta^2 \epsilon N \leq \frac{1}{8}$, 
    (3) $(1 + \eta^2 \epsilon^2)^N \leq \frac{5}{4}$, 
    (4) $(1 + \eta\epsilon)^N \geq \frac{10}{\epsilon^2}$. 
Then $x_1^{(N)}\geq \frac{1}{\epsilon^2}$ with probability at least $0.2$.
\end{lem}


\begin{proof}

Following \cite{jain2016streaming} we define the matrix
$B_k = \prod_{i=1}^k \parens{I - \eta g^{(i)}(g^{(i)})^T A}$,
where the $g^{(i)}$ are independent $\mathcal{N}(0,I)$ gaussians.  Note that $x^{(k)} = B_k x^{(0)}.$  We will show that $B_k^T e_1$ has large norm with good probability (in fact we will show that $\inner{B_k^T e_1}{e_1}$ is large). This will then imply that $\inner{B_k x^{(0)}}{e_1}$ is large with high probability, where $x^{(0)} \sim \mathcal{N}(0,I).$ 

\textbf{Step 1: Deriving a recurrence for the second moments.}

Let $y^{(k)} = B_k^T e_1$ and let $u^{(k)}_i$ be the second moment of the coordinate $y_i^{(k)}$.  Note that $u^{(0)}_i = \delta_{1i}$ (where $\delta$ is the Dirac delta).  To simplify the notation, we drop the superscript on the $g$.  We compute
    $y_i^{(k+1)} = \parens{(I - \eta gg^T A)y^{(k)}}_i
    = y^{(k)}_i - \eta (Ag)_i (g_1 y^{(k)}_1 + \ldots + g_d y^{(k)}_d)
     = y_i^{(k)} - \eta\lambda_i g_i (g_1 y^{(k)}_1 + \ldots + g_d y^{(k)}_d)$.

Next we observe that (after grouping terms) the coefficients of the $y_i^{(k)}$ terms are pairwise uncorrelated.  Using this, along with the fact that the $g_i$'s are independent of the $y_i^{(k)}$'s gives

\begin{eqnarray*}
    u_i^{(k+1)} =  \E(1-\eta\lambda_i g_i^2)^2 u_i^{(k)} + \eta^2 \lambda_i^2 \sum_{j\neq i} u_j^{(k)}
    = (1 - 2\eta\lambda_i + 3\eta^2\lambda_i^2) u_i^{(k)} + \eta^2 \lambda_i^2 \sum_{j\neq i} u_j^{(k)}\\
    = (1 - 2\eta\lambda_i + 2\eta^2\lambda_i^2) u_i^{(k)} + \eta^2 \lambda_i^2 \sum_{j=1}^d u_j^{(k)}.
\end{eqnarray*}

Let $S^{(k)} = u^{(k)}_1 + \ldots + u^{(k)}_d$, and $\gamma_i = 1 - 2\eta\lambda_i + 2\eta^2\lambda_i^2.$  Then we can write the recurrence as
    $u^{(k+1)}_i = \gamma_i u^{(k)}_i + \eta^2\lambda_i^2 S^{(k)}$.
Iterating this recurrence gives
\begin{equation}
    \label{eq:expanded_sec_mom_rec}
    u^{(k)}_i = \delta_{1i}\gamma_i^k + \eta^2 \lambda_i^2 \parens{\gamma_i^{k-1}S^{(0)} + \gamma_i^{k-2}S^{(1)} + \ldots + S^{(k-1)} }.
\end{equation}

\textbf{Step 2: Bounding $S^{(k)}$.}

Summing the above equation over $i$ allows us to write a recurrence for the $S^{(k)}$'s:
    $S^{(k)} = \gamma_1^k + \alpha_{k-1}S^{(0)} + \alpha_{k-2}S^{(1)} + \ldots + \alpha_0 S^{(k-1)}$,
where we define
    $\alpha_j := \sum_{i=1}^d \eta^2 \lambda_i^2 \gamma_i^j$.

We split $\alpha_j$ into two parts, $\alpha_j^{+}$ and $\alpha_j^{-}$ corresponding to terms in the sum where $\lambda_i$ is positive or negative respectively.
We now use the recurrence to bound $S^{(k)}.$ First by Holder's inequality, 
$S^{(k)} \leq \gamma_1^k + \max(S^{(0)}, \ldots, S^{(k-1)})(\alpha_0^{+} + \ldots + \alpha_{k-1}^{+}) 
+ (\alpha_{k-1}^{-}S^{(0)} + \alpha_{k-2}^{-}S^{(1)} + \ldots + \alpha_0^{-} S^{(k-1)}).$

We calculate 


\begin{align*}
\sum_{j=0}^{k-1} \alpha_j^{+} 
&= \sum_{j=0}^{k-1} \sum_{i : \lambda_i > 0}\eta^2 \lambda_i^2 \gamma_i^j 
= \sum_{i : \lambda_i > 0} \eta^2 \lambda_i^2 \sum_{j=0}^{k-1} \gamma_i^j 
= \sum_{i : \lambda_i > 0} \eta^2 \lambda_i^2 \frac{1-\gamma_i^k}{1 - \gamma_i}\\
&= \sum_{i : \lambda_i > 0} \eta^2 \lambda_i^2 \frac{1-\gamma_i^k}{2\eta\lambda_i - 2\eta^2\lambda_i^2}
= \sum_{i : \lambda_i > 0} \eta \lambda_i \frac{1-\gamma_i^k}{2 - 2\eta\lambda_i} \leq \sum_{i : \lambda_i > 0} \eta\lambda_i \leq \eta,
\end{align*}

where we used that $\eta \lambda_i \leq 1/2,$ (which is a consequence of Assumption 1), that $\gamma_i<1$ (which holds since $\lambda_i > 0$) and that $\sum_{i : \lambda_i > 0} \lambda_i \leq 1.$ Since we assume that $-\epsilon$ is the smallest eigenvalue,
\begin{align*}
    \alpha_j^{-} &\leq \eta^2\gamma_1^j \sum_{i:\lambda_i < 0} \lambda_i^2
    \leq \eta^2 \gamma_1^j \eps \sum_{i:\lambda_i < 0} \abs{\lambda_i} \leq \eta^2\gamma_1^j \eps.
\end{align*}  

Let $\tilde{S}^{(k)} = \max(S^{(0)}, \ldots S^{(k)}).$  Then combining our bounds gives
\begin{equation*}
    \tilde{S}^{(k)} \leq \max\parens{\tilde{S}^{(k-1)}, \gamma_1^k + \eta\tilde{S}^{(k-1)} + \eta^2\epsilon (\gamma_1^{k-1}\tilde{S}^{(0)} + \gamma_1^{k-2}\tilde{S}^{(1)} + \ldots + \tilde{S}^{(k-1)})}.
\end{equation*}

The next step is to use this recurrence to bound $\tilde{S}^{(k)}.$  For this, define $c^{(k)}$ such that $\tilde{S}^{(k)} = c^{(k)} \gamma_1^k.$ Plugging in to the above and dividing through by $\gamma_1^k$, we get that $c^{(k)}$ satisfies
\begin{align*}
    c^{(k)} &\leq \max\parens{
        \frac{c^{(k-1)}}{\gamma_1},
         1 + \frac{\eta}{\gamma_1}c^{(k-1)} + \frac{\eta^2\epsilon}{\gamma_1}(c^{(0)} + \ldots + c^{(k-1)})
    }\\
    &\leq \max\parens{
        c^{(k-1)},
        1 + \eta c^{(k-1)} + \eta^2 \epsilon(c^{(0)} + \ldots + c^{(k-1)})
    },
\end{align*}
where we used the fact that $\gamma_1\geq 1.$ Now set $\tilde{c}^{(k)} = \max(c^{(0)}, \ldots c^{(k)}).$ By assumptions 1 and 2, $\eta + \eta^2 \epsilon k \leq 1/2.$  This gives
\begin{align*}
    \tilde{c}^{(k)} \leq \max\parens{
    \tilde{c}^{(k-1)}, 1 + \eta \tilde{c}^{(k-1)}+\eta^2\epsilon k\tilde{c}^{(k-1)}
    }
    \leq \max\parens{
    \tilde{c}^{(k-1)}, 1 + \frac{1}{2}\tilde{c}^{(k-1)}
    }.
\end{align*}

Note that $c^{(0)} = S^{(0)} = 1,$ so a straightforward induction using the above recurrence shows that $\tilde{c}^{(k)} \leq 2$ for all $k$.  It follows that 
    $S^{(k)} \leq 2\gamma_1^k$.

\textbf{Step 3: Bounding the second moment.}
Plugging the bound above in to \eqref{eq:expanded_sec_mom_rec} gives
\begin{align*}
    u_1^{(k)} &\leq \gamma_1^k + 2k\eta^2 \epsilon^2 \gamma_1^{k-1}\leq \parens{1+2k\eta^2\epsilon^2}\gamma_1^k.
\end{align*}

\textbf{Step 4: Applying Chebyshev.}  We focus on the first coordinate, $y_1^{(k)}.$  Note that $I - \eta Agg^T$ has expectation $I -\eta A$, so a straightforward induction shows that 
    $\E y_1^{(k)} = (1+\eta\epsilon)^k$.

Using the bound for the second moment of the first coordinate, we get
    $\frac{u_1^{(k)}}{\parens{\E y_1^{(k)}}^2} \leq \frac{(1 + 2k\eta^2 \epsilon^2)\gamma_1^k}{(1+\eta\epsilon)^{2k}}
    = (1 + 2k\eta^2 \epsilon^2) \parens{\frac{1+2\eta\epsilon + 2\eta^2\epsilon^2}{1+2\eta\epsilon + \eta^2 \epsilon^2}}^k
    = (1 + 2k\eta^2 \epsilon^2) \parens{1 + \frac{\eta^2\epsilon^2}{1+2\eta\epsilon + \eta^2 \epsilon^2}}^k
    \leq (1 + 2k\eta^2 \epsilon^2) (1 + \eta^2\epsilon^2)^k.$

By Assumptions 2 and 4, $N \eta^2 \epsilon^2 \leq 1/8$ and $(1+\eta^2\epsilon^2)^N \leq 5/4,$ so we get that $u_1^{(k)} \leq 25/16 \parens{\E u_1^{(k)}}^2.$

Thus by Chebyshev's inequality,
    $\pr\parens{\abs{y_1^{(k)} -\E(y_1^{(k)}) } \geq 0.9 \E(y_1^{(k)})} \leq \frac{25}{36}$.
So with probability at least $0.3$,
    $y_1^{(N)} \geq \frac{1}{10} \E(y_1^{(N)}) = \frac{1}{10} (1+\eta\epsilon)^N$.

Under assumption 4, $(1+\eta\epsilon)^N\geq \frac{10}{\epsilon^2},$ which means that $y_1^{(N)}\geq \frac{1}{\epsilon^2}$ with at least $0.3$ probability.

\textbf{Step 5: Concluding the argument.}
We showed that $\inner{B_N^T e_1} {e_1} \geq \frac{1}{\epsilon^2}$ with probability at least $0.3$.  In particular this implies that $\norm{B_N^T e_1}{} \geq \frac{1}{\epsilon^2}.$  Now since $x^{(0)}$ is distributed as $\mathcal{N}(0,I),$
    $\inner{B_N x^{(0)}} {e_1} 
    = \inner{x^{(0)}}{B_N^T e_1}
    \sim \mathcal{N}(0, \norm{B_N^T e_1}{}^2)$,
which is at least $\norm{B_N^T e_1}{}$ in magnitude with $0.67$ probability. It follows that $x^{(N)}_1 \geq \frac{1}{\epsilon^2}$ with probability at least $0.2$.
\end{proof}

Let $f(x) = x^T A x.$ We next understand how the value of $f(x^{(k)})$ is updated on each iteration.
\begin{prop}
\label{prop:rayleigh_quotient_update}
For $g \sim \mathcal{N}(0,1)$, we have 
    $f(x^{(k)}) - f(x^{(k+1)}) = \eta (g^T A x^{(k)})^2 (2 - \eta g^T A g)$.
\end{prop}
\begin{proof}
Plugging in the update rule and expanding gives
\begin{align*}
    f(x^{(k+1)}) &= (x^{(k)})^T A x^{(k)} - 2\eta (g^T A x^{(k)})^2 + \eta^2 (g^T A x^{(k)})^2 g^T A g\\
    &= (x^{(k)})^T A x^{(k)} - \eta (g^TAx^{(k)})^2 (2 - \eta g^T A g),
\end{align*}
from which the proposition follows.
\end{proof}

A consequence of this update is that the sequence $f(x^{(k)})$ is almost guaranteed to be decreasing as long as $\eta$ is chosen small enough.

\begin{prop}
\label{prop:rayleigh_quotient_decreasing}
Assume that $\Tr(A) \leq 1$ and that $\eta < c$. After $N$ iterations, $f(x^{(N)}) \leq f(x^{(0)})$ with probability at least $99/100$ provided that $\eta \leq \frac{c}{\log N + 1}.$ 
\end{prop}

\begin{proof}
We show something stronger; namely that for the first $N$ iterations, the sequence $f(x^{(k)})$ is decreasing.  By Proposition \ref{prop:rayleigh_quotient_update}, $f(x^{(k+1)})\leq f(x^{(k)})$ as long as $g^T A g \leq \frac{2}{\eta}.$  The probability that this does not occur is
$\Pr\parens{\sum \lambda_i g_i^2 \geq \frac{2}{\eta}} \leq \Pr\parens{\sum \lambda_i (g_i^2-1) \geq \frac{2}{\eta} - 1}$.

The $g_i^2 - 1$ terms are independent subexponential random variables. So by Bernstein's inequality (see \cite{vershynin2018high} Theorem 2.8.2 for the version used here), this probability is bounded by $2\exp(-c/\eta)$ as long as $\eta$ is a sufficiently small constant.  Taking a union bound gives that $f(x^{(N)}) \leq f(x^{(0)})$ with probability at least $1 - 2N\exp(-c/\eta)$, which is at least $99/100$ under the conditions given.
\end{proof}

\begin{thm}
\label{thm:adaptive_analysis}
Suppose that $\norm{A}{1} \leq 1$, $\epsilon < 1/2$, and that $A$ has $-\epsilon$ as an eigenvalue. If we take $\eta$ such that
$c\eps^2\leq \eta \leq \min\parens{\frac{1}{32 \log(10/\epsilon^2)}, \frac{c}{\log\frac{1}{\epsilon}}}$,
then for some $N = \Theta\parens{\frac{1}{\epsilon\eta}\log \frac{1}{\epsilon} }$ we have $f(x^{(N)}) < 0$ with at least $1/10$ probability.
\end{thm}

\begin{proof}
Given an $\eta$ as in the statement of the theorem, choose 
$N = \left\lceil\frac{2}{\eta\epsilon} \log\frac{10}{\epsilon^2} \right\rceil,$
which satisfies the assumptions of Lemma~\ref{lem:first_coord_grows}.  Then $x^{(N)}_1 \geq \frac{1}{\epsilon^2}$ with probability at least $0.2$.  By proposition~\ref{prop:rayleigh_quotient_decreasing}, $f(x^{(N)}) \leq f(x^{(0)}) \leq 2$ with at least $0.99$  probability, using the fact that $\eta\leq \frac{c}{\log\frac{1}{\epsilon}}$ for an appropriately chosen absolute constant $c$, such that the hypothesis of Proposition~\ref{prop:rayleigh_quotient_decreasing} holds.

If $f(x^{(N)}) < 0$, then the algorithm has already terminated.  Otherwise conditioned on the events in the above paragraph, we have with at least $0.8$ probability that $2 - \eta (g^{(N)})^T A g^{(N)} \geq \frac{1}{2}$ and 
\[
(g^{(N)})^T A x^{(N)} \sim \mathcal{N}\left(0, \|Ax^{(N)}\|^2\right) \geq \frac{1}{3}\norm{Ax^{(N)}}{}
\geq \frac{1}{3}\lambda_1 x^{(N)}_1
\geq \frac{1}{3\epsilon^2}\lambda_1 \geq \frac{1}{3\epsilon}.
\] 
Then by Proposition~\ref{prop:rayleigh_quotient_update} it follows that
    $f(x^{(N+1)}) \leq f(x^{(N)}) - \frac{\eta}{20\epsilon^2} \leq 2 - \frac{\eta}{20\epsilon^2} < 0.$
\end{proof}


We also observe that we can reduce the dimension of the problem by using a result of Andoni and Nguyen.  This allows us to avoid a $\log d$ dependence.
\begin{prop}
\label{prop:andoni_sketching_result}
Suppose that $A$ satisfies $\lambda_{\min}(A) < -\alpha \norm{A}{1},$ and let $G \in \R^{d\times m}$ have independent $\mathcal{N}(0,\frac1d)$ entries.  Then we can choose $m=O(1/\alpha)$ such that $\lambda_{\min}(G^T A G) < -\alpha/2$ and $\norm{G^T A G}{1} \leq 2 \norm{A}{1}.$
\end{prop}
\begin{proof}
For the first claim, we simply apply Theorem 1.1 in \cite{andoni2013eigenvalues} and (in their notation) set $\eps = O(1)$ and $k = O(1/\alpha).$

To show that the Schatten $1$-norm does not grow too much under the sketch, we first write $A = A_{+} + A_{-}$ where the nonzero eigenvalues of $A_{+}$ are exactly the positive eigenvalues of $A$.  Then using the usual analysis of Hutchinson's trace estimator (see \cite{meyer2021hutch++} for example), we have
\begin{align*}
    \norm{G^T A G}{1} &\leq \norm{G^T A_+ G}{1} + \norm{G^T A_- G}{1} = \Tr(G^T A_+ G) + \Tr(G^T A_- G)\\ 
    &= (1\pm O(1/\sqrt{m}))\Tr(A_{+}) + (1\pm O(1/\sqrt{m}))\Tr(A_{-})\\ 
    &\leq (1\pm O(1/\sqrt{m})\norm{A}{1}.
\end{align*}

\end{proof}

\noindent We are now ready to give the proof of Theorem~\ref{thm:adaptive_vmv_upper}.
\begin{proof}
The above result applies after scaling the $\eta$ given in Theorem~\ref{thm:adaptive_analysis}
by $1/\norm{A}{1}$.  So it suffices to choose $\eta$ to be bounded above by $$\frac{1}{\norm{A}{1}}\min\parens{\frac{1}{32 \log(10/\epsilon^2)}, \frac{c}{\log\frac{1}{\epsilon}}},$$ and within a constant factor of this value. 

To choose an $\eta$, pick a standard normal $g$, and compute $Ag$ using $1/\eps$ vector-matrix-vector queries.  Then with constant probability, $\lambda_{\max}(A) \leq \norm{Ag}{} \leq 2d\lambda_{\max}.$ Given this, we have
\begin{equation}
    d\norm{Ag}{} \geq \norm{A}{1} \geq \frac{\norm{Ag}{}}{2d},
\end{equation}
which allows us to approximate $\norm{A}{1}$ to within a factor of $d^2$ with constant probability.  Given this, one may simply try the above algorithm with an $\eta$ at each of $O(\log(d^2)) = O(\log d)$ different scales, with the cost of an extra $\log d$ factor. 

Finally, we may improve the $\log d$ factor to a $\log(1/\eps)$ factor by using Proposition~\ref{prop:andoni_sketching_result} to sketch $A$, and then applying the above analysis to $G^T A G.$  Note that the sketch may be used implicitly; once $G$ is chosen, a vector-matrix-vector query to $G^T A G$ can be simulated with a single vector-matrix-vector query to $A.$
\end{proof}

\subsection{Lower bounds}
We will show a bound for two-sided testers which will imply that the bound for \ptester{1}s given in Theorem~\ref{thm:adaptive_vmv_upper} is tight up to $\log$ factors. If we require the tester to have one-sided error, then we additionally show that the bound in Corollary~\ref{cor:adaptive_vmv_upper_p} is tight for all $p$.  Note that this distinction between one-sided and two-sided testers is necessary given Theorem~\ref{thm:adaptive_l2_upper}.

In order to obtain these lower bounds for adaptive testers, we first show corresponding lower bounds for non-adaptive testers. A minor modification to Lemma 3.1 in \cite{sun2019querying} shows that an adaptive tester can have at most quadratic improvement over a non-adaptive tester.  This will allow us to obtain our adaptive lower bounds as a consequence of the non-adaptive bounds.

\subsubsection{Non-adaptive lower bounds}

We first observe that a one-sided tester must always be able to produce a witness matrix $X$, that at least certifies that $A$ is not positive definite.

\begin{prop}
\label{prop:one_sided_vmv_implies_witness}
If a one-sided tester makes a series of symmetric linear measurements $\inner{M_i}{A}$ of $A$, and outputs \False\,on a given instance, then there must exist nonzero $X\in \spn(M_1,\ldots, M_k)$ such that $X$ is PSD and $\inner{X}{A} \leq 0.$
\end{prop}

\begin{proof}

We work within the space of symmetric matrices. Let $W = \spn(M_1,\ldots, M_k)$, and let $\phi(X) = \inner{A}{X}$ be the linear functional associated with $A.$ Now suppose that $\phi$ is strictly positive for all nonzero $X \in W \cap \psdcone{d}$. We will construct $\widetilde{\phi}$ that agrees with $\phi$ on $W$ and is non-negative on $\psdcone{d}$.

Let $W' = \ker(\phi) \cap W$, and note that $W' \cap \psdcone{d} = \{0\}$.  Now by convexity of $\psdcone{d}$, there exists a hyperplane $H$ and associated half-space $H^+$ such that (i)$H$ contains $W'$ (ii) $H\cap \psdcone{d} = \{0\},$ (iii) $H^+ \supseteq \psdcone{d}$ and (iv) $\phi$ is non-negative on $H^+ \cap W$ . Moreover, since $W'$ intersects $\psdcone{d}$ trivially, $H$ can be chosen such that $H\cap W = W'.$ Now let $\Pi$ be a projection onto $W$ that maps $H$ to $W'$, and choose $\widetilde{\phi} = \phi\circ \Pi_W$.

The linear functional $\widetilde{\phi}$ is represented by the inner product against some symmetric matrix $B.$  By construction of $\widetilde{\phi}$, we have $\inner{B}{M_i} = \inner{A}{M_i}$ for all $i$, and also $\inner{B}{X} \geq 0$ for all PSD $X$. So in particular $\inner{B}{xx^T} = x^T B x \geq 0$ for all $x$, which implies that $B$ is PSD. Given the existence of the PSD matrix $B$ consistent with all measurements, the one-sided tester must not reject.



\end{proof}

We are now able to give an explicit non-PSD spectrum which is hard for any one-sided tester.  Specifically, we show that it is hard for any vector-matrix-vector query algorithm to produce a witness $X$ in the sense of the proposition above.

\begin{thm}
\label{thm:nonadaptive_explicit_hard_instance}

Let $\lambda>0$ and suppose for all matrices $A$ with spectrum $(-\lambda,1,\ldots,1)$ that a non-adaptive one-sided tester $\mathcal{T}$ outputs \False\,with $2/3$ probability.  Then $\mathcal{T}$ must make at least $\frac19 \parens{\frac{d}{1+\lambda}}^2$ vector-matrix-vector queries.

\end{thm}

\begin{proof}

By the polarization identity, $$x^T A y = \frac{1}{2}\parens{(x+y)^T A (x+y) - y^T A y - x^T A x},$$ we may assume that all queries are of the form $x_i^T A x_i,$ at the cost requiring at most a factor of three increase in the number of queries.


We set $A = I - (1+\lambda)vv^T$ where $v$ is uniform over $S^{d-1}$, and let $W = \spn(x_1x_1^T,\ldots x_kx_k^T).$  By Proposition~\ref{prop:one_sided_vmv_implies_witness}, the tester may only reject if there is an $X$ in $W \cap \psdcone{d}$ with $\norm{X}{F} = 1$ such that $\inner{X}{A} \leq 0.$ For such an $X$ we have 
\begin{equation}
    \inner{vv^T}{X} \geq \frac{\Tr(X)}{1 + \lambda} \geq \frac{1}{1+\lambda}.
\end{equation}
But since $vv^T$ and $X$ both have unit norm and $X\in W$, this condition implies that $\norm{\Pi_W(vv^T)}{F} \geq \frac{1}{1+\lambda}.$






Now we turn to understanding $\E(\norm{\Pi_W(v v^T)}{F}^2).$ Indeed we have the following:

\begin{lem}
\label{lem:identity}
Let $v$ be drawn uniformly from $S^{d-1},$ and let $W$ be a $k$-dimensional subspace of the $d\times d$ symmetric matrices.  Let $\alpha_4 = \E(v_1^4)$ and $\alpha_{22} = E(v_1^2v_2^2).$ Then $$\E(\norm{\Pi_W(v v^T)}{F}^2) = (\alpha_4 - \alpha_{22})k + \alpha_{22} \norm{\Pi_W(I)}{F}^2,$$ where $I$ is the identity matrix.
\end{lem}

\begin{proof}
Let $M_1,\ldots M_k$ be an orthonormal basis for $W$.  By the Pythagorean theorem, 
\begin{equation}
    \label{pythag}
    \E(\norm{\Pi_W(v v^T)}{F}^2) = \sum_{i=1}^k \E(\norm{\Pi_{M_i}(v v^T)}{F}^2).
\end{equation}

For fixed $M$ we have 
\begin{equation}
    \E(\norm{\Pi_M(v v^T)}{F}^2) = \E(\inner{v v^T}{M}^2) = \E(\Tr(v v^T M)^2) = \E((v^T M v)^2).\\
\end{equation}

Since $M$ is symmetric, we can diagonalize $M$ to $D = \diag(a_1, \ldots, a_d)$ in some orthonormal basis.  Since $M$ has unit norm, $a_1^2 + \ldots + a_d^2 = 1.$ Then we have
\begin{align*}
    \E((v^T M v)^2) &= \E((v^T D v)^2)\\
    &= \E( (a_1x_1^2 + \ldots + a_d x_d^2)^2  )\\
    &= \alpha_4 (a_1^2 + \ldots + a_d^2) + 2\alpha_{22} \sum_{i<j}a_ia_j\\
    &= \alpha_4 + 2\alpha_{22} \sum_{i<j}a_ia_j.
\end{align*}

Next observe that $$\Tr(M)^2 = (a_1 + \ldots + a_d)^2  = a_1^2 + \ldots + a_d^2 + 2\sum_{i<j} a_i a_j = 1 + 2\sum_{i<j} a_i a_j,$$ so that $$\E((v^T M v)^2) = \alpha_4 + \alpha_{22}(\Tr(M)^2 - 1).$$  Combining with \eqref{pythag} gives

\begin{align*}
    \E(\norm{\Pi_W(v v^T)}{F}^2) &= \sum_{i=1}^k \parens{\alpha_4 + \alpha_{22}(\Tr(M_i)^2 - 1)}\\
    &= (\alpha_4 - \alpha_{22})k + \alpha_{22} \sum_{i=1}^k \Tr(M_i)^2.
\end{align*}

Finally, observe that $\sum_{i=1}^k \Tr(M_i)^2 = \sum_{i=1}^k \inner{I}{M_i}^2 = \norm{\Pi_W(I)}{F}^2,$ by the Pythagorean theorem, which finishes the proof.
\end{proof}

\begin{rmk}
\label{rmk:unit_vect_moment_cals}
While approximations would suffice, this result gives a quick way to compute $\alpha_4$ and $\alpha_{22}.$  Set $W$ to be the entire space of $n\times n$ symmetric matrices, and $k = d(d+1)/2.$ The previous result gives $$1 = \frac{d(d+1)}{2}(\alpha_4 - \alpha_{22}) + d\alpha_{22}.$$ On the other hand, by expanding $1 = (v_1^2 + \ldots + v_d^2)^2,$ we have $$1 = d\alpha_4 + d(d-1)\alpha_{22}.$$  Solving the system yields $\alpha_4 = \frac{3}{d(d+2)}$ and $\alpha_{22} = \frac{1}{d(d+2)}.$
\end{rmk}

To finish the proof of the theorem, we recall that $W$ is spanned by the matrices $x_1 x_1^T, \ldots, x_k x_k^T,$ each of which has rank one.  Therefore each matrix in $W$, and in particular $\Pi_W(I)$, has rank at most $k$. 

We recall for a general matrix $A,$ that $\argmin_{\rank(U) \leq k}\norm{A - U}{F}$ is gotten by truncating all but the largest $k$ singular values of $A.$  Applying this to the identity matrix, when $k\leq d,$ we see that $$\norm{\Pi_W(I)}{F}^2 = \norm{I}{F}^2 - \norm{\Pi_{W^{\perp}}(I)}{F} \leq d - (d-k) = k,$$
since $\norm{\Pi_{W^{\perp}}(I)}{F}^2 \geq \min_{\rank(U)\leq k} \norm{I - U}{F}^2 = d-k.$  Since $\norm{I}{F}^2 = d,$ we always have $\norm{\Pi_W(I)}{F}^2 \leq k.$

Combining this fact with Lemma~\ref{lem:identity} gives 
\begin{equation*}
    \E(\norm{\Pi_W(v v^T)}{F}^2) \leq (\alpha_4 - \alpha_{22})k + \alpha_{22}k
    = k \alpha_4
    = \frac{3k}{d(d+2)}
    \leq \frac{3k}{d^2},
\end{equation*}
and by Markov's inequality, $$\pr\left(\norm{\Pi_W(v v^T)}{F}^2 > \frac{9k}{d^2}\right) \leq \frac{1}{3}.$$ So with probability $2/3,$ $\norm{\Pi_W(v v^T)}{F}^2 \leq \frac{9k}{d^2}.$ But for $\mathcal{A}$ to be correct, we saw that we must have $\norm{\Pi_W(v v^T)}{F} \geq \frac{1}{1+\lambda}$ with probability $2/3$.  It follows that $$\parens{\frac{1}{1+\lambda}}^2 \leq \frac{9k}{d^2},$$ which implies that $$k\geq \frac{1}{9}\parens{\frac{d}{1+\lambda}}^2.$$
\end{proof}

In particular, this result implies that for non-adaptive one-sided testers, a $\poly(1/\eps)$ \ptester{p} can only exist for $p=1.$ 

\begin{thm}
\label{thm:nonadaptive_one_sided_vmv_lower}
A one-sided non-adaptive \ptester{p} must make at least $\Omega(\frac {1}{\eps^2} d^{2-2/p})$ vector-matrix-vector queries. 
\end{thm}

\begin{proof}
This follows as a corollary of Theorem~\ref{thm:nonadaptive_explicit_hard_instance};  simply apply that result to the spectrum $(\eps (d-1)^{1/p}, 1\ldots, 1)$ where there are $d-1$ $1$'s.
\end{proof}

\subsubsection{Adaptive lower bounds}

As remarked earlier, our adaptive lower bounds follow as a corollary of our non-adaptive bounds, and a slightly modified version of Lemma 3.1 in \cite{sun2019querying}, which we give here.

\begin{lem}
\label{lem:at_most_quadratic_improvement}
Let $A = X\Sigma X^T$ be a random symmetric $d\times d$ real-valued matrix, with $\Sigma$ diagonal, and where $X$ is orthonormal and sampled from the rotationally invariant distribution.  Any $s$ adaptive vector-matrix-vector queries to $A$ may be simulated by $O(s^2)$ non-adaptive vector-matrix-vector queries.
\end{lem}

\begin{proof}
(Sketch) First note that the adaptive protocol may be simulated by $3s$ adaptive quadratic form queries, of the form $x^T A x$ by the polarization identity
\begin{equation}
    x^T A y = \frac{1}{2}\parens{(x+y)^T A (x+y) - x^T A x - y^T A y}.
\end{equation}

These queries in turn may be simulated by $9s^2$ non-adaptive queries by following exactly the same proof as Lemma 3.1 in \cite{sun2019querying} (but now with $u_i = v_i$ in their proof).
\end{proof}


As a direct consequence of this fact and our Theorem~\ref{thm:nonadaptive_one_sided_vmv_lower} we obtain the following.
\begin{thm}
\label{thm:one_sided_p_lower_bound}
An adaptive one-sided \ptester{p} must make at least $\Omega(\frac{1}{\eps}d^{1-1/p})$ vector-matrix-vector queries.
\end{thm}

\section{Adaptive matrix-vector queries}


We analyze random Krylov iteration.  Namely we begin with a random $g\sim \mathcal{N}(0,I_d)$ and construct the sequence of iterates $g, Ag, A^2 g, \ldots A^{k}g$ using $k$ adaptive matrix-vector queries. The span of these vectors is denoted $\mathcal{K}_{k}(g)$ and referred to as the ${k}^{\text{th}}$ Krylov subspace.

Krylov iteration suggests a very simple algorithm.  First compute $g, Ag, \ldots, A^{k+1}g.$ If $\mathcal{K}_{k}(g)$ contains a vector $v$ such that $v^T A v < 0$ then output \False, otherwise output \True. (Note that one can compute $Av$ and hence $v^T A v$ for all such $v$, given the $k+1$ matrix-vector queries.)  We show that this simple algorithm is in fact optimal.

\td{Maybe say a little more}
As a point of implementation, we note that the above condition on $\mathcal{K}_{k}(g)$ can be checked algorthmically.  One first uses Gram-Schmidt to compute the projection $\Pi$ onto $\mathcal{K}_{k}(g)$.  The existence of a $v\in \mathcal{K}_{k}(g)$ with $v^T A b < 0$ is equivalent to the condition $\lambda_{\min}(\Pi A \Pi) < 0$.  When $A$ is $\eps$-far from PSD, the proof below will show that in fact $\lambda_{\min}(\Pi A \Pi)<-\Omega(\eps)\norm{A}{p}$, so it suffices to estimate $\lambda_{\min}(\Pi A \Pi)$ to within $O(\eps)\norm{A}{p}$ accuracy.

\begin{prop}
\label{prop:chebyshev_construction}
For $r > 0$, $\alpha>0$ and $\delta>0$ there exists a polynomial $p$ of degree $O(\frac{\sqrt{r}}{\sqrt{\alpha}}\log\frac{1}{\delta})$, such that $p(-\alpha) = 1$ and $\abs{p(x)}\leq \delta$ for all $x\in [0,r].$
\end{prop}

\begin{proof}
Recall that the degree $d$ Chebyshev polynomial $T_d$ is bounded by $1$ in absolute value on $[-1,1]$ and satisfies \[T_d(1+\gamma \geq 2^{d\sqrt{\gamma} - 1}).\] (See \cite{musco2015randomized} for example.)   The proposition follows by shifting and scaling $T_d$.
\end{proof}

\begin{thm}
\label{thm:adaptive_mv_upper}
Suppose that $A$ has an eigenvalue $\lambda_{\min}$ with $\lambda_{\min} \leq -\eps \norm{A}{p}.$ When $p=1$, the Krylov subspace $\mathcal{K}_{k}(g)$ contains a vector $v$ with $v^T A v < 0$ for $k = O\parens{\parens{\frac1\eps}^{\frac{1}{3}} \log \frac1\eps}.$  When $p\in (1,\infty]$, the same conclusion holds for $k = O\parens{\parens{\frac1\eps}^{\frac{p}{2p+1}} \log \frac1\eps \log d}.$
\end{thm}

\begin{proof}

Without loss of generality, assume that $\norm{A}{p}\leq 1$. Fix a value $T$ to be determined later, effectively corresponding to the number of top eigenvalues that we deflate. By Proposition~\ref{prop:chebyshev_construction} we can construct a polynomial $q$, such that $q(\lambda_{\min}) = 1$ and $\abs{q(x)} \leq \sqrt{\frac{\eps/10}{d^{1 - 1/p}}}$ for $x\in [0,T^{-1/p}]$ with
\begin{equation}
    \deg(q) \leq C \frac{T^{-1/(2p)}}{\sqrt{\eps}}\log\parens{\sqrt{\frac{d^{1-1/p}}{\eps/10}}},
\end{equation} where $C$ is an absolute constant.

Now set 
\begin{equation}
    p(x) = q(x) \prod_{i : \lambda_i > T^{-1/p}} \frac{\lambda_i - x}{\lambda_i - \lambda_{\min}}.
\end{equation}
Since we assume $\norm{A}{p} \leq 1,$ there at most $T$ terms in the product, so 
\begin{equation}
    \deg(p) \leq T + C \frac{T^{-1/(2p)}}{\sqrt{\eps}}\log\parens{\sqrt{\frac{d^{1-1/p}}{\eps/10}}}.
\end{equation}
By setting $T = \eps^{- p/(2p+1)},$ we get
\begin{equation}
\deg(p) = 
    \begin{cases} 
        O\parens{\parens{\frac1\eps}^{\frac{p}{2p+1}} \log \frac1\eps} \quad &\text{if } p=1\\
        O\parens{\parens{\frac1\eps}^{\frac{p}{2p+1}} \log \frac1\eps \log d} \quad &\text{if } p>1
    \end{cases}
\end{equation}
As long as $k$ is at least $\deg(p)$, then $v = p(A)g$ lies in $\mathcal{K}_k(g),$ and
\begin{equation}
    v^T A v = g^T p(A)^2 A g.
\end{equation}
By construction, $p(\lambda_{\min}) = 1$.  Also for all $x$ in $[0, T^{-1/p}],$ $|p(x)| \leq |q(x)| \leq \sqrt{\eps/10} d^{(1/p) - 1}.$

Therefore the matrix $p(A)^2 A$ has at least one eigenvalue less than $-\eps$, and the positive eigenvalues sum to at most
\begin{equation}
    \sum_{i:\lambda_i>0} \frac{\eps}{10} d^{1/p - 1} \lambda_i \leq \frac{\eps}{10},
\end{equation} by using Holder's inequality along with the fact that $\norm{A}{p} \leq 1$. So with at least $2/3$ probability, $g^T p(A)^2 A g < 0$ as desired.
\end{proof}



\begin{rmk}
For $1<p<\infty$, the $\log d$ dependence can be removed by simply applying the $p=1$ tester to  $A^{\lceil p \rceil},$ as a matrix-vector query to $A^{\lceil p \rceil}$ may be simulated via $\lceil p \rceil$ matrix-vector queries to $A.$  However this comes at the cost of a $\parens{\frac1\eps}^{\lceil p \rceil/3}$ dependence, and is therefore only an improvement when $d$ is extremely large.
\end{rmk}

\begin{rmk}
While we observe that deflation of the top eigenvalues can be carried out implicitly within the Krylov space, this can also be done explicitly using block Krylov iteration, along with the guarantee given in Theorem 1 of \cite{musco2015randomized}.
\end{rmk}




We showed above that we could improve upon the usual analysis of Krylov iteration in our context. We next establish a matching lower bound that shows our analysis is tight up to $\log$ factors.  This is a corollary of the proof of Theorem 3.1 presented in \cite{braverman2020gradient}.

\begin{thm}
\label{thm:adaptive_mv_lower}
A two-sided, adaptive \ptester{p} in the matrix-vector model must in general make at least $\Omega(\frac{1}{\epsilon^{p/(2p+1)}})$ queries.
\end{thm}

\begin{proof}
We make use of the proof of Theorem 3.1 given in \cite{braverman2020gradient}.  We consider an algorithm $\mathcal{A}$ that receives a matrix $W$ sampled from the Wishart distribution makes at most $(1-\beta)d$ queries, and outputs either True or False, depending on whether $\lambda_{\min}(W)$ is greater or less than $t$ (where $t=1/(2d^2)$ is defined as in \cite{braverman2020gradient}). We say that $\mathcal{A}$ fails on a given instance if either (i) $\mathcal{A}$ outputs True and $t - \frac{1}{4d^2} \geq \lambda_{\min}(W)$ or (ii) $\mathcal{A}$ outputs False and $\lambda_{\min}(W) \geq t + \frac{1}{4d^2}.$ Exactly the same proof given in \cite{braverman2020gradient} shows that $\mathcal{A}$ must fail with probability at least $c_{\text{wish}}\sqrt{\beta}$ where $c_{\text{wish}}>0$ is an absolute constant, as long as $d$ is chosen sufficiently large depending on $\beta.$ Taking $\beta = 1/4$ say, means that any such algorithm fails with probability at least $c_{\text{wish}}/2$ as long as $d$ is a large enough constant.

Now consider an \ptester{p} $\mathcal{T}$ with $d = 1/(4\epsilon^{p/(2p+1)})$, applied to the random matrix $W - tI.$  While our definition allows $\mathcal{T}$ to fail with $1/3$ probability we can reduce this failure probability to $c_{\text{wish}}/2$ by running a constant number of independent instances and taking a majority vote. So from here on we assume that $\mathcal{T}$ fails on a given instance with probability at most $c_{\text{wish}}/2$.

First recall that $W \sim XX^T$ where each entry of $X$ is i.i.d. $\mathcal{N}(0,1/d).$ Then with high probability, the operator norm of $X$ is bounded, say, by $\sqrt{2}$, and the eigenvalues of $W$ are bounded by $2.$


Therefore with high probability, $\norm{W}{p} \leq 2d^{1/p},$ and so $\norm{W - tI}{p} \leq 3d^{1/p}.$  It follows that $1/(4d^2) = 4\epsilon (4d)^{1/p} \geq \epsilon \norm{W-tI}{p}.$ This means that $\mathcal{T}$ can solve the problem above, and by correctness of the tester, fails with at most $c_{\text{wish}}/2$ probability. For $\epsilon$ sufficiently small, the above analysis implies that $\mathcal{T}$ must make at least $\Omega(d) = \Omega(1/\epsilon^{p/(2p+1)})$ queries.

\end{proof}


\section{An optimal bilinear sketch}


In this section we analyze a bilinear sketch for PSD-testing which will also yield an optimal  \ptester{2} in the vector-matrix-vector model.

Our sketch is very simple.  We choose $G \in \R^{d\times k}$ to have independent $\mathcal{N}(0,1)$ entries and take our sketch to be $G^T A G.$  In parallel we construct estimates $\alpha$ and $\beta$ for the trace and Frobenius norm of $A$ respectively, such that $\beta$ is accurate to within a multiplicative error of $2$, and $\alpha$ is accurate to with $\norm{A}{F}$ additive error. (Note that this may be done at the cost of increasing the sketching dimension by $O(1)$.)

If $G^T A G$ is not PSD then we automatically reject. Otherwise, we then consider the quantity
\begin{equation}
    \gamma := \frac{\alpha - \lambda_{\min}(G^TAG)}{\beta \sqrt{k}\log k}
\end{equation}
If $\gamma$ is at most $c_{\text{psd}}$ for some absolute constant $c_{\text{psd}}$, then the tester outputs \False, otherwise it outputs \True.

We first show that a large negative eigenvalue of $A$ results causes the smallest sketched eigenvalue to be at most $\Tr(A) - \Omega(\eps k).$  On the other hand, when $A$ is PSD, we will show that $\lambda_{\min}(G^TAG)$ is substantially larger.

\subsection{Upper bound on $\lambda_{\min}(G^T A G)$}

We start with the following result on trace estimators which we will need below.
\begin{prop}
\label{prop:variance_calculation}
Let $M$ be a symmetric matrix with eigenvalues $\lambda_1, \ldots, \lambda_d,$ and let $u$ be a random unit vector with respect to the spherical measure.  Then
\[\var(u^T M u) = \frac{2}{d+2} \parens{\frac{\lambda_1^2 + \ldots + \lambda_d^2}{d} - \frac{(\lambda_1 + \ldots + \lambda_d)^2}{d^2}} := \frac{2}{d+2} \var(\lambda_1, \ldots, \lambda_d).\]
\end{prop}

\begin{proof}
By the spectral theorem, it suffices to prove the result when $M$ is diagonal.  Then
\[\var{u^T M u} = \E(\lambda_1 u_1^2 +\ldots + \lambda_d  u_d^2)^2 - \parens{\E(\lambda_1 u_1^2 + \ldots + \lambda_d u_d^2)}^2.\]

By Remark~\ref{rmk:unit_vect_moment_cals}, we have $\E(u_i^2) = 1,$ $\E(u_i^4) = \frac{3}{d(d+2)}$ and $\E(u_i^2 u_j ^2) = \frac{1}{d(d+2)}$ for $i\neq j.$  The result follows by expanding using linearity of expectation, and then applying these facts.
\end{proof}

The next two results will give an upper bound on the smallest eigenvalue of the Gaussian sketch.  For the proof of Lemma~\ref{lem:sketch_eval_upper_bd} we will start with random orthogonal projections, from which the Gaussian result will quickly follow. We include a technical hypothesis that essentially enforces non-negativity of $\Tr(A).$  We write the hypothesis in the form below simply to streamline the argument.

\begin{lem}
\label{lem:rand_projection_sketch}
Suppose that $\norm{A}{F}=1$ and that $v$ is an eigenvector of $A$ with associated eigenvalue $-\eps.$ Let $\Pi \in R^{d\times d}$ be a projection onto a random $k$ dimensional subspace $S$ of $\R^d$, sampled from the rotationally invariant measure.  Also suppose that $x^T A x \geq 0$ with probability $0.999$ when $x$ is a random unit vector.  Then
\[\frac{1}{\norm{\Pi v}{}^2} (\Pi v)^T A (\Pi v) \leq \frac{1}{d} (-0.5 \eps k + \Tr(A) + O(1))\] 
with probability at least $0.99 - \exp(-ck)$.
\end{lem}

\begin{proof}
Let $u = \frac{\Pi v}{\norm{\Pi v}{}}.$ The subspace $S$ was chosen randomly, so with probability at least $1-\exp(-ck)$, 
\[\inner{u}{v}^2 = \norm{\Pi v}{}^2 \geq 0.5 \frac{k}{d}.\]
\td{citation?} 
Let $u'$ be the projection of $u$ onto the hyperplane $v^{\perp}$ orthogonal to $v.$  Observe by symmetry that $u'/\norm{u'}{}$ is distributed uniformly over the sphere in $v^{\perp}$.

Let $A' = A - \eps vv^T$ be the matrix $A$ with the $-\eps$ eigenvalue zeroed out.  Then 
\[u^T A u \leq -0.5 \frac{k}{d} \eps + (u')^T A' u'\leq -0.5 \frac{k}{d} \eps + (u'/\norm{u'}{})^T A' (u'/\norm{u'}{}),\]
as long as $(u')^T A' u' \geq 0$, which holds with probability at least $0.999$ as a consequence of the similar hypothesis.
The latter term is a trace estimator for $\frac{1}{d}A'$ with variance bounded by $\frac{2}{d^2}\norm{A'}{F}^2 \leq \frac{2}{d^2}$ (for example by Proposition~\ref{prop:variance_calculation}). So with $0.999$ probability
\[(u'/\norm{u'}{})^T A' (u'/\norm{u'}{}) \leq \frac{1}{d}(\Tr(A') + O(1)) = \frac{1}{d}(\Tr(A) + \eps + O(1)) \leq \frac{1}{d}(\Tr(A) + O(1)),\] 
and the result follows.
\end{proof}

In the following lemma, we introduce the technical assumption that $k < cd$. However this will be unimportant later, as any sketch with $k\geq cd$ might as well have sketching dimension $d$, at which point the testing problem is trivial.

\begin{lem}
\label{lem:sketch_eval_upper_bd}
Suppose that $A\in \R^{d\times d}$ has an eigenvalue of $-\eps$, $\norm{A}{F}=1$, and $G \in \R^{d\times k}$ with $k<d$ has iid $\mathcal{N}(0,1)$ entries.  Also suppose that $x^T A x\geq 0$ with probability at least $0.999$ for a random unit vector $x,$ and that $k<cd$ for some absolute constant $c.$  Then with probability at least $0.99 - 3\exp(-ck)$,
\[\lambda_{\min}(G^T A G) \leq  - 0.4\eps k + O(1) + \Tr(A) + c\frac{\sqrt{k}}{\sqrt{d}}\abs{\Tr(A)} .\]

\end{lem}


\begin{proof}
Let $\Pi_G = GG^{\pinv}$ denote projection onto the image of $G$.

Let $v$ be an eigenvector of $A$ with associated eigenvalue smaller or equal to $-\eps,$ and set $u = G^{\pinv}v/\norm{G^{\pinv}v}{}$.
We then have 
\[\lambda_{\min}(G^T A G) \leq u^T (G^T A G) u = \frac{1}{\norm{G^{\pinv}v}{}^2} (\Pi_G v)^T A (\Pi_G v).
\]
We also have 
\[\norm{\Pi_G v}{} = \norm{G G^{\pinv} v}{} \leq \norm{G}{\text{op}} \norm{G^{\pinv} v}{}.\]
By Theorem 4.6.1 in \cite{vershynin2018high}, $\norm{G}{\text{op}} \leq \sqrt{d} + c\sqrt{k}$ with probability at least $1 - 2\exp(-k)$. Conditional on this occurring, 
\[\norm{G^{\pinv} v}{} \geq \frac{\norm{\Pi_G v}{}}{\sqrt{d} + c \sqrt{k}},\]
from which it follows that
\[\lambda_{\min}(G^T A G) \leq (d + c\sqrt{d}\sqrt{k}) \frac{1}{\norm{\Pi_G v}{}^2} (\Pi_G v)^T A (\Pi_G v),\] as long as the quantity on the right-hand side is non-negative.  If this quantity is negative, then we similarly have
\[\lambda_{\min}(G^T A G) \leq (d - c\sqrt{d}\sqrt{k}) \frac{1}{\norm{\Pi_G v}{}^2} (\Pi_G v)^T A (\Pi_G v),\]
using the analogous bound on the smallest singular value of $G.$

Since $G$ is Gaussian, the image of $G$ is distributed with the respect to the rotationally invariant measure on $k$-dimensional subspaces. Therefore Lemma~\ref{lem:rand_projection_sketch} applies, and the result follows after collecting terms, and using the assumption that $k\leq cd$ in the negative case.
\end{proof}




\subsection{Lower bound on $\lambda_{\min}(G^T A G)$}

We follow a standard protocol for bounding the extreme eigenvalues of a random matrix.  We first show that $u^T (G^T A G) u$ is reasonably large for a fixed vector $u$ with high probability.  Then by taking a union bound over an $\eps$-net we upgrade this to a uniform bound over the sphere.

We require two additional tricks.  Our lower bound on $u^T (G^T A G) u$ arises from Berstein's inequality, which is hampered by the existence of large eigenvalues of $A$.  Therefore in order to get a guarantee that holds with high enough probability, we first prune the large eigenvalues of $A.$ 

Second, the mesh size of our $\eps$-net needs to be inversely proportional to the Lipschitz constant of $x\mapsto x^T (G^T A G) x$ as $x$ ranges over the sphere.  A priori, the Lipschitz constant might be as bad as $\norm{A}{\text{op}}$ which is typically larger than $\sqrt{d}$. This would ultimately give rise to an additional $\log(d) $ factor in the final sketching dimension.  However we show that the the Lipschitz constant is in fact bounded by $O(k)$ with good probability, avoiding the need for any $d$ dependence in the sketching dimension.

\begin{prop}
\label{prop:quad_form_inequality}
Let $Q$ be a symmetric matrix, and let $x$ and $y$ be unit vectors.  Then \[\abs{x^T Q x - y^T Q y} \leq 2 (\lambda_{\max}(Q) - \lambda_{\min}(Q))\norm{x-y}{}\]
\end{prop}
\begin{proof}
We first reduce to the $2$-dimensional case.  Let $W$ be a $2$-dimensional subspace passing through $x$ and $y$.  The largest and smallest eigenvalues of the restriction to $W$ of the quadratic form associated to $Q$ are bounded from above and below by $\lambda_{\max}(Q)$ and $\lambda_{\min}(Q)$ respectively.  It therefore suffices to prove the result when $Q$ has dimension $2.$

Since the result we wish to show is invariant under shifting $Q$ by multiples of the identity, it suffices to consider the case when $\lambda_2 = 0.$  After these reductions, we have \[x^T Q x - y^T Q y = \lambda_1 (x_1^2 - y_1^2) = \lambda_1 (x_1 + y_1)(x_1 - y_1).\]  Since $x$ and $y$ are unit vectors, $\abs{x_1 + y_1} \leq 2$ and $\abs{x_1 - y_1} \leq \norm{x-y}{}$ and the result follows.
\end{proof}

\begin{lem}
\label{lem:avoid_logd_dependence}
Let $S = G^T A G$ where $G\in \R^{k\times d}$ has iid $\mathcal{N}(0,1)$ entries and $\norm{A}{F} = 1.$  Then \[\lambda_{\max}(S) - \lambda_{\min}(S)  \leq t\] with probability at least $1 - \frac{4k(k+2)}{t^2}$.
\end{lem}

\begin{proof}

Consider the random quantity $\alpha = u^T G^T A G u$, where $u$ is a random unit vector in $\R^k$, independent from $G$.  Note that $Gu$ is distributed as a standard Gaussian, so $\alpha$ is a trace estimator for $A$ with variance $2$ \cite{avron2011randomized}.

On the other hand, one can also study the variance of $\alpha$ conditional on $G$ by using Proposition~\ref{prop:variance_calculation}. Let $E$ be the event that $\lambda_{\max}(S) - \lambda_{\min}(S)  \geq t$.  If $E$ occurs too often, then $\var(\alpha)$ would be too large. Specifically, in the notation of  Proposition~\ref{prop:variance_calculation} when $E$ occurs, we necessarily have $\var(\lambda_1(S), \ldots, \lambda_k(S)) \geq \frac{1}{k} (t/2)^2$, so $\var(\alpha | E) \geq \frac{t^2}{2k(k+2)}.$ Thus we have
\begin{equation}
    2 = \var(\alpha) \geq  \Pr(E)  \var(\alpha | E)  \geq  \Pr(E)    \frac{t^2}{2k(k+2)},
\end{equation}
and so $\Pr(E) \leq 4k(k+2)/t^2$ as desired.
\end{proof}

\begin{lem}
\label{lem:single_point_in_net_bound}
Suppose that $A$ is PSD with $\norm{A}{F}=1,$ and that $v$ consists of iid $\mathcal{N}(0,1)$ entries.  Then for $t \geq 2\sqrt{k}$ we have
\[
\Pr(v^T A v \leq \Tr(A) - t) \leq \exp(-c\sqrt{k}(t-\sqrt{k})).
\]
\end{lem}

\begin{proof}
We have 
\begin{align}
    \Pr(v^T A v \leq \Tr(A) - t) &\leq \Pr(v^T A_{-k} v \leq \Tr(A) - t)\\
    &= \Pr(v^T A_{-k} v \leq \Tr(A_{-k}) -(t- \Tr(A_{k}))
\end{align}

Note that $v^T A_{-k} v$ has expectation $\Tr(A_{-k})$, so by Bernstein's inequality  (or Hanson-Wright) \cite{vershynin2018high}, 
\[\Pr(v^T A v \leq \Tr(A) - t) \leq \exp\parens{-c \min\parens{\frac{(t-\Tr(A_{k}))^2}{\norm{A_{-k}}{F}^2} , \frac{t - \Tr(A_{k})}{\lambda_{\max}(A_{-k})} } },\]
for $t\geq \Tr(A_k).$

Now note that $\norm{A_{-k}}{F} \leq \norm{A}{F} \leq 1$, and that $\lambda_{\max}(A_{-k}) \leq \frac{1}{\sqrt{k}},$ since $\norm{A}{F}=1.$ Additionally, $\Tr(A_k) \leq \sqrt{k} \norm{A_k}{F} \leq \sqrt{k}.$  These bounds imply that
\[\frac{(t-\Tr(A_{k}))^2}{\norm{A_{-k}}{F}^2} \geq (t-\sqrt{k})^2\]
and
\[\frac{t - \Tr(A_{k})}{\lambda_{\max}(A_{-k})} \geq \sqrt{k}(t-\sqrt{k})\]
for $t\geq \Tr(A_k).$  When $t > 2\sqrt{k}$, the latter expression is smaller, and the conclusion follows.

\end{proof}


\begin{thm}
\label{thm:sketch_eval_lower_bd}
Suppose that $A$ is PSD with $\norm{A}{F}=1$, and that $G\in \R^{d\times k}$ has iid $\mathcal{N}(0,1)$ entries and that $k\geq 5$. Then with at least $0.99$ probability,
\[
\lambda_{\min}(G^T A G) \geq \Tr(A) - c\sqrt{k}\log(k)
\] for some absolute constant $c$.
\end{thm}

\begin{proof}



For any fixed unit vector $u\in \R^k$, $Gu$ is distributed as a standard Gaussian, and so  Lemma~\ref{lem:single_point_in_net_bound} applies.  Therefore for a choice of constant, 

\begin{equation}
\label{eq:bound_over_net}
u^T (G^T A G) u \geq \Tr(A) - c\sqrt{k}\log(k)
\end{equation}
with probability at least $1 - \exp(-10k\log k).$

Let $\mathcal{N}$ be a net for the sphere in $\R^k$ with mesh size $1/k,$ which can be taken to have at most $(3k)^k$ elements \cite{vershynin2018high}.  By taking a union bound, equation~\ref{eq:bound_over_net} holds over $\mathcal{N}$ with probability at least \[1 - (3k)^k\exp(-10k\log k) \geq 1 - \exp(-k)\] for $k\geq 2.$

By choosing $t=100k$ in Lemma~\ref{lem:avoid_logd_dependence}, and applying Proposition~\ref{prop:quad_form_inequality}, we get that \[\abs{x^T (G^T A G) x - y^T (G^T A G)y} \leq 100k\norm{x-y}{}\] with probability at least $0.999$.  Since $\mathcal{N}$ has mesh size $1/k$, we have that
\[u^T (G^T A G) u \geq \Tr(A) - c\sqrt{k} + O(1)\]
for all unit vectors $u$ in $\R^k$ with probability at least $0.999 - \exp(-k).$




\end{proof}

As a consequence of Theorem~\ref{thm:sketch_eval_lower_bd} and Lemma~\ref{lem:sketch_eval_upper_bd} we obtain our main result.
\begin{thm}
\label{thm:bilinear_sketch_guarantee}
There is a bilinear sketch $G^T A G$ with sketching dimension $k = O(\frac{1}{\eps^2} \log^2\frac{1}{\eps})$ that yields a two-sided \ptester{2} that is correct with at least $0.9$ probability.
\end{thm}
\begin{proof}
If $\lambda_{\min}(A) < 0$ then we automatically reject. Otherwise we first use $O(1)$ columns of the sketching matrices to estimate $\alpha$ of $\Tr(A)$ to within an additive error of $\norm{A}{F}$ with $0.01$ failure probability. We then use another $O(1)$ columns of the sketching matrices to construct an approximation $\beta$ of $\norm{A}{F}$ with $\frac{1}{2} \norm{A}{F} \leq \beta \leq 2 \norm{A}{F}$ with $0.01$ failure probability (see for example \cite{meister2019tight}).

Now consider the quantity 
\[\gamma = \frac{\alpha - \lambda_{\min}(G^T A G)}{\beta\sqrt{k}\log k}.\]

If $A$ is PSD, then by Theorem~\ref{thm:sketch_eval_lower_bd}, 
\[\lambda_{\min}(G^T A G) \geq \Tr(A) - c\sqrt{k}\log k \norm{A}{F},\]
which implies that
\[\gamma \leq c_{\text{psd}},\]
for some absolute constant $c_{\text{psd}}.$

On the other hand, if $A$ has a negative eigenvalue less than or equal to $-\eps$, then by Theorem~\ref{lem:sketch_eval_upper_bd}
\[ \lambda_{\min}(G^T A G) \leq - 0.4\eps k \norm{A}{F} + \parens{1 + c\frac{\sqrt{k}}{\sqrt{d}}}\Tr(A) + O(1)\norm{A}{F},\]
which implies that 
\[\gamma \geq c_{\text{far}}(\eps \sqrt{k} - c)/\log k,\] 
for some absolute constant $c_{\text{far}}.$

Finally by taking $k = \Theta(\frac{1}{\eps^2} \log^2\frac{1}{\eps})$, we have $c_{\text{psd}} < c_{\text{far}}(\eps \sqrt{k}-c)/\log k$, which implies that the tester is correct if it outputs \True\,precisely when $\gamma \leq c_{\text{psd}}.$

\end{proof}

Note that this result immediately gives a non-adaptive vector-matrix-vector tester which makes $\tildeO(1/\eps^4)$ queries.


\subsection{Application to adaptive vector-matrix-vector queries}

By combining our bilinear sketch with Theorem~\ref{thm:adaptive_vmv_upper} we achieve tight bounds for adaptive queries.

\begin{thm}
\label{thm:adaptive_l2_upper}
There is a two-sided adaptive \ptester{2} in the vector-matrix-vector model, which makes $\tildeO(1/\eps^2)$ queries.
\end{thm}

\begin{proof}
To handle the technical condition in Lemma~\ref{lem:sketch_eval_upper_bd}, we first compute $x^T A x$ for a constant number of independent Gaussian vectors $x.$  If $x^T A x$ is ever negative, then we automatically reject.

We showed in the proof of Theorem~\ref{thm:bilinear_sketch_guarantee} that with $0.9$ probability, $\gamma \leq c_{psd}$ if $A$ is PSD, and $\gamma \geq c_{\text{far}}(\eps\sqrt{k} - c)/\log k := C_{\text{far}}(k)$ if $A$ is $\eps$-far from PSD.  By choosing some $k = O(\frac{1}{\eps^2} \log^2 \frac{1}{\eps})$ we can arrange for $C_{\text{far}}(k) - c_{\text{psd}} \geq 1,$ and also for $C_{\text{far}}(k) = \Theta(1).$

Next we compute estimates $\alpha$ and $\beta$  of $\Tr(A)$ and $\norm{A}{F}$ as above, and (implicitly) form the matrix
\[\Gamma = \frac{G^TAG - \alpha I_k}{\beta \sqrt{k}\log k} + C_{\text{far}}(k) I_k - I_k.\]

If $A$ is PSD, then with very good probability, \[\lambda_{\text{min}}(\Gamma) = -\gamma + C_{\text{far}}(k) - \frac{1}{\eps} \geq -c_{\text{psd}} + C_{\text{far}}(k) - 1 \geq 0.\]

Similarly, if $A$ is $\eps$-far from PSD, then \[\lambda_{\text{min}}(\Gamma) = -\gamma + C_{\text{far}}(k) - \frac{1}{\eps} \leq -C_{\text{far}}(k) + C_{\text{far}}(k) - 1 = -1.\]

Thus it suffices to distinguish $\Gamma$ being PSD from $\Gamma$ having a negative eigenvalue less than or equal to $-1.$

For this we will utilize our adaptive $\ell_1$-tester, so we must bound $\norm{\Gamma}{1}.$ Note that $G^T A G$ is a trace estimator for $kA$ with variance $O(k)\norm{A}{F}$ \cite{andoni2013eigenvalues}.  Therefore $\Tr(G^T A G) = k(\Tr(A) \pm O(1)\norm{A}{F}).$  Define $M = \frac{1}{\beta}(G^T A G - \alpha I_k)$, so that $\Tr(M) = O(k).$  The negative eigenvalues of $M$ sum to at most $k\lambda_{\min}(M)$ in magnitude, and so the bound on $\Tr(M)$ implies that $\norm{M}{1} \leq 2k\lambda_{\min}(M) + O(k).$  Write $\Gamma = \frac{1}{\sqrt{k}\log k} M + C_{\text{far}}(k) I - I$, so that $\norm{\Gamma}{1} \leq \frac{1}{\sqrt{k}\log k}\norm{M}{1} + O(k).$  Note that $\lambda_{\min}(M) \leq \sqrt{k}\log k \lambda_{\min}(\Gamma) + O(\sqrt{k}\log k).$  Therefore $\norm{\Gamma}{1} \leq 2k \lambda_{\min}(\Gamma) + O(k).$ From this we have
\[\frac{1}{\lambda_{\min}(\Gamma)}\norm{\Gamma}{1} \leq 2k\parens{\frac{\lambda_{\min}(\Gamma) + O(1)}{\lambda_{\min}(\Gamma)}} + \frac{O(k)}{\lambda_{\min}(\Gamma)}
\leq O(k)\] 
as long as $|\lambda_{\min}(\Gamma)| \geq \Omega(1)$, which it is by assumption.

Therefore Theorem~\ref{thm:adaptive_vmv_upper} gives an adaptive vector-matrix-vector tester for $\Gamma$ which requires only $\tildeO(k) = \tildeO(\frac{1}{\eps^2})$ queries.

\end{proof}

\noindent As a consequence we also obtain a two-sided $p$-tester for all $p\geq 2.$
\begin{corollary}
\label{cor:adaptive_lp_upper}
For $p\geq 2$, there is a two-sided adaptive \ptester{p} in the vector-matrix-vector model, which make $\tildeO(1/\eps^2) d^{1-1/p}$ queries.
\end{corollary}

\begin{proof}
Apply Theorem~\ref{thm:adaptive_l2_upper} along with the bound $\norm{A}{p} \geq d^{\frac{1}{p} - \frac{1}{2}}\norm{A}{F}.$
\end{proof}

\subsection{Lower bounds for two-sided testers}

Our lower bounds for two-sided testers come from the spiked Gaussian model introduced in \cite{li2016tight}.  As before, our adaptive lower bounds will come as a consequence of the corresponding non-adaptive bounds.

\begin{thm}
\label{thm:nonadaptive_two_sided_vmv_lower_bound}
A two-sided \ptester{p} that makes non-adaptive vector-matrix-vector queries requires at least
\begin{itemize}
    \item $\Omega(\frac{1}{\eps^{2p}})$ queries for $1\leq p \leq 2$
    \item $\Omega(\frac{1}{\eps^4} d^{2 - 4/p})$ queries for $2 < p < \infty$ as long as $d$ can be taken to be $\Omega(1/\eps^p).$
    \item $\Omega(d^2)$ queries for $p=\infty.$
\end{itemize}
\end{thm}

\begin{proof}
First, take $G$ to be a $d\times d$ matrix with $\mathcal{N}(0,1)$ entries, where $d = 1/\epsilon$. Also let $\widetilde{G} = G + suv^T$ where $u$ and $v$ have $\mathcal{N}(0,1)$ entries, and $s$ is to be chosen later.  We will show that a PSD-tester can be used to distinguish $G$ and $\widetilde{G}$, while this is hard for any algorithm that uses only a linear sketch.

Recall that $G$ has spectral norm at most $c\sqrt{d}$ with probability at least $1-2e^{-d},$ where $c$ is an absolute constant. (We will use $c$ throughout to indicate absolute constants that we do not track -- it may have different values between uses, even within the same equation.) Set 
\begin{equation}
    G_{\text{sym}} = \begin{bmatrix}
                        0 & G \\
                        G^T & 0,
\end{bmatrix}
\end{equation}
and define $\widetilde{G}_{\text{sym}}$ similarly. Note that the eigenvalues of $G_{\text{sym}}$ are precisely $\{\pm \sigma_i\}$ where the $\sigma_i$ are singular values of $G.$  Therefore $G_{\text{sym}} + c\sqrt{d} I$ is PSD with high probability.

On the other hand, $\norm{uv^T}{} \geq cd$ with high probability so
\begin{equation}
\norm{\widetilde{G}}{} \geq s\norm{uv^T}{} - \norm{G}{} \geq csd - c\sqrt{d},
\end{equation}
which implies that $\widetilde{G}_{\text{sym}} + c\sqrt{d}I$ has a negative eigenvalue with magnitude at least $csd - c\sqrt{d} - c\sqrt{d} = csd - c\sqrt{d}.$  

We also have that $\norm{\widetilde{G}_{\text{sym}}}{p} \leq c\sqrt{d}d^{1/p},$ since the operator norm of $G$ is bounded by $c\sqrt{d}$ with high probability. Hence if
\begin{equation}
    \label{eq:cond_for_two_sided_lower_bound}
    c \frac{sd - \sqrt{d}}{\sqrt{d}d^{1/p}} = c\frac{s\sqrt{d} - 1}{d^{1/p}} \geq \eps,
\end{equation} then a two-sided PSD-tester can distinguish between $G$ and $\widetilde{G}$ with constant probability of failure.

On the other hand, Theorem 4 in \cite{li2016tight} implies that any sketch that distinguishes these distributions with constant probability, must have sketching dimension at least $c/s^4.$

It remains to choose values of $s$ and $d$ for which the inequality in equation~\eqref{eq:cond_for_two_sided_lower_bound} holds. When $1\leq p \leq 2$, we take $d = \Theta(\eps^{-p})$ and $s = O(\eps^{p/2})$ giving a lower bound of $\Omega(1/\eps^{2p}).$ When $2 < p < \infty$, we take $d=\Omega(1/\eps^p)$ and $s = O(\eps d^{1/p - 1/2})$ giving a lower bound of $\Omega(\frac{1}{\eps^4} d^{2 - 4/p}).$ Finally, when $p = \infty$ we take $s=O(1/\sqrt{d})$ giving a lower bound of $\Omega(d^2).$

\end{proof}

\begin{rmk}
The argument above applies equally well to arbitrary linear sketches, of which a series of non-adaptive vector-matrix-vector queries is a special case.
\end{rmk}

\begin{corollary}
\label{cor:adaptive_two_sided_lower_bound}
A two-sided adaptive \ptester{p} in the vector-matrix-vector model requires at least
\begin{itemize}
    \item $\Omega(\frac{1}{\eps^{p}})$ queries for $1\leq p \leq 2$
    \item $\Omega(\frac{1}{\eps^2} d^{1 - 2/p})$ queries for $2 < p < \infty$ as long as $d$ can be taken to be $\Omega(1/\eps^p).$
    \item $\Omega(d)$ queries for $p=\infty.$
\end{itemize}
\end{corollary}
\begin{proof}
Apply Theorem~\ref{thm:nonadaptive_two_sided_vmv_lower_bound} along with Lemma~\ref{lem:at_most_quadratic_improvement}.
\end{proof}

For adaptive measurements, we supply a second proof via communication complexity which has the advantage of applying to general linear measurements, albeit at the cost of an additional bit complexity term.

\begin{prop}
Let $p \in [1,\infty)$, $\eps<\frac{1}{2}$ and $d \geq (p/\eps)^p$. An adaptive two-sided \ptester{p} taking general linear measurements $\inner{M_i}{A}$ of $A$, where each $M_i$ has integer entries in $(-2^{b-1}, 2^{b-1}]$, must make at least $\frac{c}{b + d\log\frac{1}{\eps}} \frac{1}{\eps^2} d^{1-2/p}$ queries.
\end{prop}

\begin{proof}


We reduce from the multiplayer set disjointness problem \cite{kamath2021simple}.  Let the $k$ players have sets $S_1, \ldots, S_k \subseteq [d]$ which either are (i) all pairwise disjoint or (ii) all  share precisely one common element.  Distinguishing between (i) and (ii) with $2/3$ probability in the blackboard model of communication requires $\Omega(d/k)$ bits.  We will choose $k = \left\lceil \max(4\eps d^{1/p}, 4p)\right\rceil = \lceil 4\eps d^{1/p} \rceil.$

For each $i$ let $\chi_i$ be the characteristic vector of $S_i$ and let $A_i = \diag(\chi_i).$  Consider the matrix $A := I_d - \sum_{i=1}^d A_i$.

In situation (i), $A$ is PSD. In situation (ii), $\norm{A}{p}^p \leq k^p + d $ and $\lambda_{\min}(A) = -(k-1).$ We have
\[\abs{\lambda_{\min}(A)}^p = (k-1)^p = k^p\parens{1-\frac1k}^p \geq k^p \parens{1-\frac{p}{k}}\geq \frac{3}{4}k^p\]
and
\[\eps^p\norm{A}{p}^p \leq \eps^p (k^p + d) = \eps^p k^p + \eps^p d \leq \frac12 k^p + \eps^p d \leq \frac{3}{4}k^p. \]

Given query access to $A$, an \ptester{p} can therefore distinguish between (i) and (ii) with $2/3$ probability.  Note that a single linear measurement $\inner{M}{A}$ may be simulated in the blackboard model using $O(k(\log b + \log d))$ bits; each player simply computes and communicates $\inner{M}{A_i}$, and the players add the resulting measurements. The players therefore need at least $\Omega(\frac{1}{k(\log b + \log d)}\cdot \frac{d}{k})$ bits of communication to solve the PSD-testing problem.


\end{proof}

\section{Spectrum Estimation}






We make use of the following result, which is Lemma 11 of \cite{clarkson2017lowPSD} specialized to our setting.

\begin{lem}
\label{lem:clarkson_woodruff_psd_result}
For a symmetric matrix $A\in \R^{d\times d}$, there is a distribution over an oblivious sketching matrix $R \in \R^{d\times m}$ with $m = O(\frac{k}{\eps})$ so that with at least $0.9$ proability,
\begin{equation}
    \min_{Y^* \in\, \textup{rank}\, k, \textup{PSD}}\norm{(AR)Y^*(AR)^T - A}{F}^2 \leq (1 + \eps)\norm{A_{k,+} - A}{F}^2,
\end{equation}
where $A_{k,+}$ is the optimal rank-one PSD approximation to $A$ in Frobenius norm.
\end{lem}


\begin{rmk}
In our setting one can simply take $R$ to be Gaussian since the guarantee above must hold when $A$ is drawn from a rotationally invariant distribution.  In many situations, structured or sparse matrices are useful, but we do not need this here.
\end{rmk}

We also recall the notion of an affine embedding \cite{clarkson2017lowReg}.  
\begin{definition}
$S$ is an affine embedding for matrices $A$ and $B$ if for all matrices $X$ of the appropriate dimensions, we have 
\begin{equation}\label{eq:theguarantee}
    \norm{S(AX - B)}{F}^2 = (1\pm \eps)\norm{AX-B}{F}^2.
\end{equation}
\end{definition}
We also recall that when $A$ is promised to have rank at most $r$, there is a distribution over $S$ with $O(\eps^{-2}r)$ rows such that \eqref{eq:theguarantee} holds with constant probability for any choice of $A$ and $B$ \cite{clarkson2017lowReg}.

\begin{lem}
\label{lem:spectral_est_lem}
There is an algorithm which makes $O(\frac{k^2}{\eps^6}\log\frac{1}{\delta})$ vector-matrix-vector queries to $A$ and with at least $1-\delta$ probability outputs an approximation of $\norm{A}{k,+}$, accurate to within $\eps \norm{A}{F}^2$ additive error.
\end{lem}
\begin{proof}

We run two subroutines in parallel.
\\
\\
\textbf{Subroutine 1. Approximate $\norm{A_{k,+} - A}{F}^2$ up to $O(\eps)$ multiplicative error.}

Our algorithm first draws affine embedding matrices $S_1$ and $S_2$ for $r=k/\eps$, and with $\eps$ distortion, each with $O(\frac{k}{\eps^3})$ rows.  We also draw a matrix $R$ as in Lemma~\ref{lem:clarkson_woodruff_psd_result} with $m = O(\frac{k}{\eps})$ columns.

We then compute $S_1 A R$ and $S_2 A R$, each requiring $\frac{k^2}{\eps^4}$ vector-matrix-vector queries, and compute $S_1 A S_2^T$ requiring $\frac{k^2}{\eps^6}$ queries. 

Let $Y_k$ be arbitrary with the appropriate dimensions (later we will optimize $Y_k$ over rank $k$ PSD matrices). By using the affine embedding property along with the fact that $R$ has rank at most $\frac{k}{\eps}$, we have
\begin{align*}
    \norm{(S_1 A R) Y_k (S_2 A R)^T + S_1 A S_2^T}{F}^2 &= (1\pm \eps)\norm{ARY_k (S_2 A R)^T + A S_2^T}{F}^2\\
    &= (1\pm \eps)\norm{S_2 A R Y_k R^T A + S_2 A}{F}^2\\
    &= (1\pm 3\eps)\norm{ARY_k R^T A + A}{F}^2.
\end{align*}
As a consequence of this, and the property held by $R$, we have
\begin{align}
    \min_{\rank(Y_k) \leq k, Y_k\,\text{PSD}}\norm{(S_1 A R) Y_k (S_2 A R)^T + S_1 A S_2^T}{F}^2 &=
    (1\pm 3\eps)\min_{Y_k}\norm{AR Y_k R^T A + A}{F}^2\\
    &= (1 \pm 7\eps) \norm{A_{k,+} - A}{F}^2.
\end{align}
Thus by computing the quantity in the left-hand-side above, our algorithm computes an $O(\eps)$ multiplicative approximation using $O(k^2/\eps^6)$ vector-matrix-vector queries.
\\
\\
\textbf{Subroutine 2. Approximate $\norm{A}{F}^2$ up to $O(\eps)$ multiplicative error.}

We simply apply Theorem 2.2. of \cite{meister2019tight}, set $q=2,$ and note that the entries of the sketch correspond to vector-matrix-vector products.  By their bound we require $O(\eps^{-2}\log(1/\eps))$ vector-matrix-vector queries.
\\
\\
Since $\norm{A_{k,+}}{F}^2 = \norm{A}{F}^2 - \norm{A_{k,+} - A}{F}^2$, we obtain an additive $O(\eps)\norm{A}{F}^2$ approximation to $\norm{A_{k,+}}{F}^2$ by running the two subroutines above and subtracting their results.

Finally, by repeating the above procedure $O(\log\frac{1}{\delta})$ times in parallel and taking the median of the trails, we obtain a failure probability of at most $\delta.$




\end{proof}

The matrices $S_1 A R$ and $S_2 A R$ in Subroutine 1 each have rank $k/\eps$ whereas the dimensions of $S_1 A S_2^T$ are $k/\eps^3.$  The matrix $S_1 A S_2^T$ therefore contains a large amount of data that will not play a role when optimizing over $Y_k$.  If $S_1 A R$ and $S_2 A R$ were known ahead of time, then we could choose to compute only the portion of $S_1 A S_2^T$ that is relevant to the optimization step, and simply estimate the Frobenius error incurred by the rest.  This allows us to construct a slightly more efficient two-pass protocol.

\begin{prop}
\label{prop:adaptive_improvement_spectral_estimation}
By using a single round adaptivity, the guarantee of Lemma~\ref{lem:spectral_est_lem} may be achieved using $O(\frac{k^2}{\eps^4} \log\frac{1}{\delta})$ vector-matrix-vector queries. 
\end{prop}
\begin{proof}

As described above, we modify Subroutine 1.  Write $M_i$ for $S_i A R$ and $Q$ for $S_1 A S_2^T.$ Instead of computing $M_1$, $M_2,$ and $Q$ at once, we instead compute $M_1$ and $M_2$ first using $k^2/\eps^4$ vector-matrix-vector queries.

We wish to estimate $\min_{Y_k} \norm{M_1 Y_k M_2^T - Q}{F}^2$, where the minimum is over PSD matrices $Y_k$ of rank at most $k$.  Let $\Pi_i$ denote orthogonal projection onto the image of $M_i$, and set $\Pi_i^{\perp} = I - \Pi_i.$  Then for fixed $Y$, we use the Pythagorean theorem to write
\begin{align}
    \norm{M_1 Y M_2 - Q}{F}^2 &= \norm{\Pi_1M_1 Y M_2\Pi_2 - Q}{F}^2\\
    &= \norm{\Pi_1(M_1 Y M_2^T - Q)\Pi_2 + \Pi_1^{\perp}Q\Pi_2 + \Pi_1 Q \Pi_2^{\perp} + \Pi_1^{\perp}Q\Pi_2^{\perp}}{F}^2\\
    &= \norm{\Pi_1(M_1 Y M_2^T - Q)\Pi_2}{F}^2 + \norm{\Pi_1^{\perp}Q\Pi_2}{F}^2 + 
    \norm{\Pi_1 Q \Pi_2^{\perp}}{F}^2 + \norm{\Pi_1^{\perp}Q\Pi_2^{\perp}}{F}^2\\
    &= \norm{M_1 Y M_2^T - \Pi_1Q\Pi_2}{F}^2 + \norm{\Pi_1^{\perp}Q\Pi_2}{F}^2 + 
    \norm{\Pi_1 Q \Pi_2^{\perp}}{F}^2 + \norm{\Pi_1^{\perp}Q\Pi_2^{\perp}}{F}^2.
\end{align}
Note that each of the last three terms can be estimated to within $O(\eps)$ multiplicative error using Subroutine 2, since a vector-matrix-vector query to one of these matrices may be simulated with a single query to $A$.  Also since each $M_i$ has rank $O(k/\eps)$, the $\Pi_i$'s are projections onto $O(k/\eps)$ dimensional subspaces. Since the $\Pi_i$'s are known to the algorithm, we may compute $\Pi_1 Q \Pi_2$ explicitly using $k^2/\eps^2$ vector-matrix-vector queries, as it suffices to query $\Pi_1 Q \Pi_2$ over the Cartesian product of bases for the images of $\Pi_1$ and $\Pi_2.$   By optimizing the first term over $Y$, we thus obtain an $O(\eps)$ multiplicative approximation to $\min_{Y_k} \norm{M_1 Y_k M_2^T - Q}{F}^2$ as desired. This gives a version of Subroutine 1 that makes $O(k^2/\eps^4)$ queries.
\end{proof}

We note that we immediately obtain a $\poly(1/\eps)$ query \ptester{2} by applying Lemma~\ref{lem:spectral_est_lem} to approximate $A_{1,-}$.  However this yields a worse $\eps$ dependence than Theorem~\ref{thm:bilinear_sketch_guarantee}. Perhaps more interestingly, these techniques also give a way to approximate the top $k$ (in magnitude) eigenvalues of $A$ while preserving their signs. We note a minor caveat. If $\lambda_k$ and $\lambda_{k+1}$ are very close in magnitude, but have opposite signs, then we cannot guarantee that we approximate $\lambda_k$.  Therefore in the statement below, we only promise to approximate eigenvalues with magnitude at least $|\lambda_k| + 2\eps$.

\td{Check the factors of 2.}
\begin{thm}
\label{thm:spectral_approximation}
Let $\lambda_1, \lambda_2,\ldots$ be the (signed) eigenvalues of $A$ sorted in decreasing order of magnitude. 

There is an algorithm that makes $O(\frac{k^2}{\eps^{12}}\log k)$ non-adaptive vector-matrix-vector queries to $A$, and with probability at least $0.9$, outputs $\tilde{\lambda}_1, \ldots, \tilde{\lambda}_k$ such that

(i)  There exists a permutation $\sigma$ on $[k]$ so that for all $i$ with $|\lambda_i| \geq |\lambda_k| + 2\eps$, $|\tilde{\lambda}_{\sigma(i)} - \lambda_i| \leq \eps \norm{A}{F}$

(ii) For all $i$, there exists $j$ with $|\lambda_j|\geq |\lambda_k| - \eps $ and $|\tilde{\lambda}_{i} - \lambda_j| \leq \eps \norm{A}{F}$

\noindent With one additional round of adaptivity the number of measurements can be reduced to $O(\frac{k^2}{\eps^{8}}\log k).$

\end{thm}

\begin{proof}
We set $\delta=\frac{1}{20k}$ in Lemma~\ref{lem:spectral_est_lem} and use it to approximate $\norm{A_{1,+}}{F}^2,\ldots, \norm{A_{k,+}}{F}^2$, along with $\norm{A_{1,-}}{F}^2,\ldots, \norm{A_{k,-}}{F}^2$, each to within $(\eps^2/2) \norm{A}{F}^2$ additive error.  Note that we may use the same sketching matrices for each of these $2k$ tasks, and then take a union bound to obtain a failure probability of at most $0.1$.  Thus we require only $O(\frac{k^2}{\eps^{12}}\log k)$ queries in total. With an additional round of adaptivity, Proposition~\ref{prop:adaptive_improvement_spectral_estimation} reduces this bound to $O(\frac{k^2}{\eps^{8}}\log k).$

Let $\lambda_{i,+}$ be the $i^{\text{th}}$ largest positive eigenvalue of $A$ if it exists, and $0$ otherwise. Define $\lambda_{i,-}$ similarly.  Note that $\lambda_{i,+}^2 = \norm{A_{i,+}}{F}^2 - \norm{A_{i-1,+}}{F}^2$ for $i\geq 2$, and that $\lambda_{1,+}^2 = \norm{A_{1,+}}{F}^2 $.  This allows us to compute approximations $\tilde{\lambda}_{i,+} \geq 0$ such that  $|\tilde{\lambda}_{i,+}^2 - \lambda_{i,+}^2| \leq \eps^2 \norm{A}{F}^2$, and similarly for the $\lambda_{i,-}$'s with $\tilde{\lambda}_{i,-} \leq 0$.  Note that this bound implies $|\tilde{\lambda}_{i,+} - \lambda_{i,+}| \leq \eps\norm{A}{F}.$

Our algorithm then simply returns the $k$ largest magnitude elements of $\{\tilde{\lambda}_{1,+}, \ldots, \tilde{\lambda}_{k,+},\tilde{\lambda}_{1,-},\ldots, \tilde{\lambda}_{k,-}\}.$
\end{proof}

\section{Non-adaptive testers}
\subsection{Non-adaptive vector-matrix-vector queries}

We gave a lower bound for one-sided testers earlier in Theorem~\ref{thm:nonadaptive_one_sided_vmv_lower}.  Here we observe that the sketch of Andoni and Nguyen \cite{andoni2013eigenvalues} provides a matching upper bound.

\begin{prop}
\label{andoni_nguyen_corollary}
There is a one-sided non-adaptive \ptester{1} that makes $O(1/\epsilon^2)$ non-adaptive vector matrix-vector queries to $A$.
\end{prop}

\begin{proof}
We simply apply Proposition~\ref{prop:andoni_sketching_result} Note that the sketch is of the form $G^TAG$, where $G\in \R^{m\times d}$ with $m = O(1/\eps)$ in our case.  Each entry of $G^TAG$ of which there are $m^2$ can be computed with a single vector-matrix-vector query.
\end{proof}

\begin{corollary}
\label{cor:andoni_nguyen_corollary_p}
There is a one-sided non-adaptive \ptester{p} that makes $O(\frac{1}{\epsilon^2} d^{2 - 2/p})$ non-adaptive vector matrix-vector queries to $A$.
\end{corollary}

\begin{proof}
Apply the previous proposition along with the bound $\norm{A}{p} \geq d^{1/p - 1} \norm{A}{1}.$
\end{proof}




\subsection{Non-adaptive matrix-vector queries}
As a simple corollary of the algorithm given by Corollary~\ref{cor:andoni_nguyen_corollary_p} we have the following.
\begin{prop}
\label{prop:nonadaptive_one_sided_mv_upper}
There exists a one-sided non-adaptive tester making $O(\frac1\eps d^{1 - 1/p})$ matrix-vector queries.
\end{prop}
\begin{proof}
Simply note that a $k\times k$ bilinear sketch may be simulated with $k$ matrix-vector queries.
\end{proof}

We next show that this bound is tight. While we consider the case where the tester queries the standard basis vectors, this is done essentially without loss of generality as any non-adaptive tester may be implemented by querying on an orthonormal set.

\begin{prop}
\label{prop:matrix_vector_witness}
Suppose that a one sided matrix-vector tester queries on the standard basis vectors $e_1,\ldots, e_k$ and outputs \False.  Let $U$ be the top $k\times k$ submatrix of $[Ae_1, \ldots , Ae_k].$ Then if $U$ is non-singular, there must exist a ``witness vector" $v\in \spn(x_1,\ldots, x_k)$ such that $v^T A v < 0$.
\end{prop}

\begin{proof}
Let $Q$ be the matrix with columns $Ae_i,$ and decompose it as 
\begin{equation}
    Q = \begin{pmatrix}
        U\\
        B
    \end{pmatrix}^T
\end{equation} where $U\in \R^{k\times k}$ and $B\in \R^{d\times (d-k)}.$ Suppose that there does not exist a $v$ as in the statement of the proposition. Note that this implies that $U$ is PSD, and in fact positive definite by the assumption that $U$ was non-singular. Now consider the block matrix 
\begin{equation}
\widetilde{A}s = \begin{pmatrix}
U & B^T \\
B & \lambda I
\end{pmatrix}
\end{equation}
for some choice of $\lambda>0.$ For arbitrary $v$ and $w$ of the appropriate dimensions, we have

\begin{align}
\begin{pmatrix}
v & w
\end{pmatrix}
\begin{pmatrix}
U & B^T \\
B & \lambda I
\end{pmatrix}
\begin{pmatrix}
v \\
w
\end{pmatrix}
&= v^T U v + 2v^T B w + \lambda w^T w\\
&\geq \norm{v}{}^2 \sigma_{\min}(U) + \lambda \norm{w}{}^2 - 2\norm{v}{}\norm{w}{}\sigma_{\max}(B).
\end{align}
Since $\sigma_{\min}(U)\neq 0$ this expression viewed as a quadratic form in $\norm{v}{}$ and $\norm{w}{}$ is positive definite for large enough $\lambda.$  This implies that $\widetilde{A}$ is positive definite as well.  Since $\widetilde{A} e_i = A e_i$ by construction, this shows that the queries are consistent with a PSD matrix.  So a one-sided tester that cannot produce a witness vector in this case must not output \False.
\end{proof}

\begin{thm}
\label{thm:non_adaptive_mv_general_lower}
Set $D = \diag(-\lambda, 1, \ldots, 1),$ let $S$ be a random orthogonal matrix, and take $A = S^T D S$. In the matrix-vector model, a one-sided non-adaptive tester must make at least $\frac{1}{2}\frac{d}{1+\lambda}$ queries to be correct on this distribution with $2/3$ probability.
\end{thm}

\begin{proof}
Given this distribution we may assume without loss of generality that the tester queries on $e_1,\ldots e_k,$ whose span we call $E_k.$ Let $u$ denote the $-\lambda$ eigen-direction of $A,$ which is distributed uniformly over $S^{d-1}.$ For unit vectors $x$, the quadratic form associated to $A$ is negative exactly when $\inner{x}{u}^2 \geq \frac{1}{1+\lambda}.$  Also the $U$ as in Proposition~\ref{prop:matrix_vector_witness} is non-singular with probability $1.$ In this case, by Proposition~\ref{prop:matrix_vector_witness} the tester can only succeed if $\norm{\Pi_{E_k}u}{}^2 \geq \frac{1}{1+\lambda}.$ On the other hand $\E \norm{\Pi_{E_k}u}{}^2 = k/d$, so by Markov, $\norm{\Pi_{E_k}u}{}^2 \leq 2 k/d$ with probability at least $1/2.$ Therefore a tester that succeeds with $2/3$ probability must have $2k/d \geq 1/(1+\lambda)$.
\end{proof}

\begin{corollary}
\label{cor:mv_one_sided_non_adaptive_lower}
In the matrix-vector model, a one-sided non-adaptive \ptester{p} must make at least $\Omega(\frac1\eps d^{1-1/p})$ queries.
\end{corollary}
\begin{proof}
Apply Theorem~\ref{thm:non_adaptive_mv_general_lower} with $\lambda = \eps d^{1/p}.$
\end{proof}

\section{Conclusion and Open Problems}

\td{Modify these}

We gave a series of tight bounds for PSD-testing in both the matrix-vector and vector-matrix-vector models. We provided tight bounds as well as a separation between one and two-sided testers in the latter model.   There are a number of additional questions that may yield interesting future work.

\begin{itemize}
\item Our adaptive vector-matrix-vector algorithm for $p=1$ uses $\Omega(1/\eps)$ rounds of adaptivity, but this may not always be desirable in practice, since the queries cannot be run in parallel.  Are there good algorithms that use less adaptivity?  What is the optimal trade-off between query complexity and the number of rounds of adaptivity?


\item One could modify our testing model and consider testers which should output $\False$\,whenever the $\ell_p$ norm of the negative eigenvalues is at least an $\epsilon$ fraction of the $\ell_p$ norm of positive eigenvalues.  Is it possible to give tight bounds for this problem in the models that we considered?

\item Is it possible to use the ideas behind our two-sided bilinear sketch to give better bounds for spectral estimation with additive Frobenius error? 
\end{itemize}

\section{Acknowledgements}
The authors would like to thank Sushant Sachdeva for pointing out that our analysis for matrix-vector queries was similar to the application of \cite{axelsson1986rate} in \cite{spielman2009note}.
\printbibliography

@inproceedings{meyer2021hutch++,
  title={Hutch++: Optimal Stochastic Trace Estimation},
  author={Meyer, Raphael A and Musco, Cameron and Musco, Christopher and Woodruff, David P},
  booktitle={Symposium on Simplicity in Algorithms (SOSA)},
  pages={142--155},
  year={2021},
  organization={SIAM}
}

@inproceedings{wimmer2014optimal,
  title={Optimal query complexity for estimating the trace of a matrix},
  author={Wimmer, Karl and Wu, Yi and Zhang, Peng},
  booktitle={International Colloquium on Automata, Languages, and Programming},
  pages={1051--1062},
  year={2014},
  organization={Springer}
}

@inproceedings{ks03,
  author    = {Robert Krauthgamer and
               Ori Sasson},
  title     = {Property testing of data dimensionality},
  booktitle = {Proceedings of the Fourteenth Annual {ACM-SIAM} Symposium on Discrete
               Algorithms, January 12-14, 2003, Baltimore, Maryland, {USA}},
  pages     = {18--27},
  publisher = {{ACM/SIAM}},
  year      = {2003},
}

@article{bakshi2020testing,
  title={Testing positive semi-definiteness via random submatrices},
  author={Bakshi, Ainesh and Chepurko, Nadiia and Jayaram, Rajesh},
  journal={arXiv preprint arXiv:2005.06441},
  year={2020}
}

@inproceedings{li2016tight,
  title={Tight bounds for sketching the operator norm, Schatten norms, and subspace embeddings},
  author={Li, Yi and Woodruff, David P},
  booktitle={Approximation, Randomization, and Combinatorial Optimization. Algorithms and Techniques (APPROX/RANDOM 2016)},
  year={2016},
  organization={Schloss Dagstuhl-Leibniz-Zentrum fuer Informatik}
}

@inproceedings{balcan2019testing,
  title={Testing matrix rank, optimally},
  author={Balcan, Maria-Florina and Li, Yi and Woodruff, David P and Zhang, Hongyang},
  booktitle={Proceedings of the Thirtieth Annual ACM-SIAM Symposium on Discrete Algorithms},
  pages={727--746},
  year={2019},
  organization={SIAM}
}

@article{han2017approximating,
  title={Approximating spectral sums of large-scale matrices using stochastic chebyshev approximations},
  author={Han, Insu and Malioutov, Dmitry and Avron, Haim and Shin, Jinwoo},
  journal={SIAM Journal on Scientific Computing},
  volume={39},
  number={4},
  pages={A1558--A1585},
  year={2017},
  publisher={SIAM}
}

@article{bai1996some,
  title={Some large-scale matrix computation problems},
  author={Bai, Zhaojun and Fahey, Gark and Golub, Gene},
  journal={Journal of Computational and Applied Mathematics},
  volume={74},
  number={1-2},
  pages={71--89},
  year={1996},
  publisher={Elsevier}
}

@article{bmr21,
  author    = {Rajarshi Bhattacharjee and
               Cameron Musco and
               Archan Ray},
  title     = {Sublinear Time Eigenvalue Approximation via Random Sampling},
  journal   = {CoRR},
  volume    = {abs/2109.07647},
  year      = {2021}
}

@article{hutchinson1989stochastic,
  title={A stochastic estimator of the trace of the influence matrix for Laplacian smoothing splines},
  author={Hutchinson, Michael F},
  journal={Communications in Statistics-Simulation and Computation},
  volume={18},
  number={3},
  pages={1059--1076},
  year={1989},
  publisher={Taylor \& Francis}
}

@article{sun2019querying,
  title={Querying a matrix through matrix-vector products},
  author={Sun, Xiaoming and Woodruff, David P and Yang, Guang and Zhang, Jialin},
  journal={arXiv preprint arXiv:1906.05736},
  year={2019}
}

@article{avron2011randomized,
  title={Randomized algorithms for estimating the trace of an implicit symmetric positive semi-definite matrix},
  author={Avron, Haim and Toledo, Sivan},
  journal={Journal of the ACM (JACM)},
  volume={58},
  number={2},
  pages={1--34},
  year={2011},
  publisher={ACM New York, NY, USA}
}

@article{rashtchian2020vector,
  title={Vector-matrix-vector queries for solving linear algebra, statistics, and graph problems},
  author={Rashtchian, Cyrus and Woodruff, David P and Zhu, Hanlin},
  journal={arXiv preprint arXiv:2006.14015},
  year={2020}
}

@article{bhattacharjee2021sublinear,
  title={Sublinear Time Eigenvalue Approximation via Random Sampling},
  author={Bhattacharjee, Rajarshi and Musco, Cameron and Ray, Archan},
  journal={arXiv preprint arXiv:2109.07647},
  year={2021}
}

@inproceedings{braverman2020gradient,
  title={The gradient complexity of linear regression},
  author={Braverman, Mark and Hazan, Elad and Simchowitz, Max and Woodworth, Blake},
  booktitle={Conference on Learning Theory},
  pages={627--647},
  year={2020},
  organization={PMLR}
}

@inproceedings{simchowitz2018tight,
  title={Tight query complexity lower bounds for PCA via finite sample deformed Wigner law},
  author={Simchowitz, Max and El Alaoui, Ahmed and Recht, Benjamin},
  booktitle={Proceedings of the 50th Annual ACM SIGACT Symposium on Theory of Computing},
  pages={1249--1259},
  year={2018}
}

@inproceedings{shamir2016convergence,
  title={Convergence of stochastic gradient descent for PCA},
  author={Shamir, Ohad},
  booktitle={International Conference on Machine Learning},
  pages={257--265},
  year={2016},
  organization={PMLR}
}

@inproceedings{inw22, 
  title = {Frequency Estimation with One-Sided Error},
  author = {Indyk, Piotr and Narayanan, Shyam and Woodruff, David P.},
  booktitle = {SODA},
  year = {2022}
  }

@article{musco2015randomized,
  title={Randomized block krylov methods for stronger and faster approximate singular value decomposition},
  author={Musco, Cameron and Musco, Christopher},
  journal={arXiv preprint arXiv:1504.05477},
  year={2015}
}

@inproceedings{andoni2013eigenvalues,
  title={Eigenvalues of a matrix in the streaming model},
  author={Andoni, Alexandr and Nguyễn, Huy L},
  booktitle={Proceedings of the twenty-fourth annual ACM-SIAM symposium on Discrete algorithms},
  pages={1729--1737},
  year={2013},
  organization={SIAM}
}

@inproceedings{jain2016streaming,
  title={Streaming pca: Matching matrix bernstein and near-optimal finite sample guarantees for oja’s algorithm},
  author={Jain, Prateek and Jin, Chi and Kakade, Sham M and Netrapalli, Praneeth and Sidford, Aaron},
  booktitle={Conference on learning theory},
  pages={1147--1164},
  year={2016},
  organization={PMLR}
}

@inproceedings{allen2017first,
  title={First efficient convergence for streaming k-pca: a global, gap-free, and near-optimal rate},
  author={Allen-Zhu, Zeyuan and Li, Yuanzhi},
  booktitle={2017 IEEE 58th Annual Symposium on Foundations of Computer Science (FOCS)},
  pages={487--492},
  year={2017},
  organization={IEEE}
}

@book{vershynin2018high,
  title={High-dimensional probability: An introduction with applications in data science},
  author={Vershynin, Roman},
  volume={47},
  year={2018},
  publisher={Cambridge university press}
}

@inproceedings{clarkson2017lowPSD,
  title={Low-rank PSD approximation in input-sparsity time},
  author={Clarkson, Kenneth L and Woodruff, David P},
  booktitle={Proceedings of the Twenty-Eighth Annual ACM-SIAM Symposium on Discrete Algorithms},
  pages={2061--2072},
  year={2017},
  organization={SIAM}
}

@article{clarkson2017lowReg,
  title={Low-rank approximation and regression in input sparsity time},
  author={Clarkson, Kenneth L and Woodruff, David P},
  journal={Journal of the ACM (JACM)},
  volume={63},
  number={6},
  pages={1--45},
  year={2017},
  publisher={ACM New York, NY, USA}
}

@article{meister2019tight,
  title={Tight dimensionality reduction for sketching low degree polynomial kernels},
  author={Meister, Michela and Sarlos, Tamas and Woodruff, David},
  year={2019}
}

@article{oja1982simplified,
  title={Simplified neuron model as a principal component analyzer},
  author={Oja, Erkki},
  journal={Journal of mathematical biology},
  volume={15},
  number={3},
  pages={267--273},
  year={1982},
  publisher={Springer}
}

@article{litvak2005smallest,
  title={Smallest singular value of random matrices and geometry of random polytopes},
  author={Litvak, Alexander E and Pajor, Alain and Rudelson, Mark and Tomczak-Jaegermann, Nicole},
  journal={Advances in Mathematics},
  volume={195},
  number={2},
  pages={491--523},
  year={2005},
  publisher={Elsevier}
}

@inproceedings{ai2016new,
  title={New characterizations in turnstile streams with applications},
  author={Ai, Yuqing and Hu, Wei and Li, Yi and Woodruff, David P},
  booktitle={31st Conference on Computational Complexity (CCC 2016)},
  year={2016},
  organization={Schloss Dagstuhl-Leibniz-Zentrum fuer Informatik}
}

@article{kamath2021simple,
  title={A simple proof of a new set disjointness with applications to data streams},
  author={Kamath, Akshay and Price, Eric and Woodruff, David P},
  journal={arXiv preprint arXiv:2105.11338},
  year={2021}
}

@article{axelsson1986rate,
  title={On the rate of convergence of the preconditioned conjugate gradient method},
  author={Axelsson, Owe and Lindskog, Gunhild},
  journal={Numerische Mathematik},
  volume={48},
  pages={499--523},
  year={1986},
  publisher={Springer}
}

@article{spielman2009note,
  title={A note on preconditioning by low-stretch spanning trees},
  author={Spielman, Daniel A and Woo, Jaeoh},
  journal={arXiv preprint arXiv:0903.2816},
  year={2009}
}
\end{document}